\documentclass[10pt,letter]{article}
\usepackage{tikz}
\usetikzlibrary{shapes,decorations,shadows}
\usetikzlibrary{decorations.pathmorphing}
\usetikzlibrary{decorations.shapes}
\usetikzlibrary{fadings}
\usetikzlibrary{patterns}
\usetikzlibrary{calc}
\usetikzlibrary{decorations.text}
\usetikzlibrary{decorations.footprints}
\usetikzlibrary{decorations.fractals}
\usetikzlibrary{shapes.gates.logic.IEC}
\usetikzlibrary{shapes.gates.logic.US}
\usetikzlibrary{fit,chains}
\usetikzlibrary{positioning}
\usepgflibrary{shapes}
\usetikzlibrary{scopes}
\usepackage{color,pgf,tikz}
\usepackage{animate}
\usetikzlibrary{calc,external,arrows}
\usetikzlibrary{backgrounds,fit,decorations.pathreplacing}  

\usepackage{graphicx}
\usepackage{verbatim}
\usepackage{subcaption}
\usepackage[linesnumbered,ruled,norelsize]{algorithm2e}

\usepackage{pgfplots}
\usepackage{setspace}
\usepackage{filecontents}
\pgfplotsset{compat=newest}
\usepackage[hidelinks]{hyperref}

\usepackage{fancyhdr}  
\usepackage[top=1in, bottom=1in, left=1in, right=1in]{geometry}
\usepackage[tableposition=top]{caption}
\usepackage{longtable}
\usepackage{booktabs}
\usepackage{pdflscape}
\usepackage{latexsym}
\usepackage{amsmath}
\usepackage{amsfonts}
\usepackage{amssymb}
\usepackage{amsthm}
\usepackage{multirow}
\usepackage{float}
\usepackage[noabbrev]{cleveref}
\usepackage{sidecap}
\usepackage{multibib}
\usepackage{cite}
\usepackage{textcomp}
\usepackage{wrapfig}
\usepackage{nomencl}
\usepackage{breakurl,url}
\usepackage{mathrsfs}  
\hypersetup{colorlinks=false,urlbordercolor=white,linkbordercolor = white,allbordercolors=white} 

\theoremstyle{remark}

\newtheorem{thm}{Theorem}[section]

\newtheorem{lem}[thm]{Lemma}

\newtheorem{prop}[thm]{Proposition}

\newtheorem{cor}{Corollary}

\newtheorem{defn}{Definition}[section]

\newtheorem{exmp}{Example}[section]

\newtheorem{rem}{Remark}

\usepackage{enumerate}

\newcommand{\exclose}{\hfill$\blacksquare$} 

\newcommand{\mv}[1] {{\rm{#1}}} 
\newcommand{\tv}[1]{{\mathbf{#1}}} 
\newcommand{\sv}[1]{{\mathcal{#1}}} 

\newcommand{\tmv}[1]{{\widetilde{\rm{#1}}}} 

\newcommand{\ttv}[1]{{\widetilde{\mathbf{#1}}}} 

\newcommand{\whsv}[1]{{\widehat{\mathcal{#1}}}} 

\usepackage{color}
\usepackage{cite}
\usepackage{cleveref}
\usepackage{setspace}
\date{}
\begin{document}

\title{A Linear Approach to Fault Analysis and Intervention in Boolean Systems}

\author{Anuj~Deshpande~and~Ritwik~Kumar~Layek 
\thanks{The authors are with the Department of Electronics and Electrical Communication Engineering, Indian Institute of Technology - Kharagpur, Kharagpur, WB, 721302 India. e-mail: deshpande.anuj@ece.iitkgp.ernet.in; ritwik@ece.iitkgp.ernet.in}
}


\maketitle
\section*{Abstract}
The mutations of a complex systemic disease like cancer can be modeled as stuck-at faults in the Boolean system paradigm. For a class of multiple faults, the fault identification is exceptionally significant under the incomplete access of all the underlying proteins of the system. A comprehensive linear framework has been developed in this manuscript to identify the class of faults under a set of homeostatic input conditions. An algorithm is developed to design new reporters to improve the observability. The other aspect of this manuscript lies in controlling the manifestation of the mutations, which is the essential objective of systems medicine research. The primary goal is to synthesize a cocktail of drug molecules (combination therapy) from a set of existing targeted drugs. The controllability results are included in this paper to understand the problem formally. An improvement of controllability algorithm is discussed to design new target drugs if the available drugs fail to accommodate the underlying fault set. The results are presented for Boolean maps and Boolean control networks. Biological examples are given to highlight the relevant results. \\

\textbf{Keywords}: Boolean control networks, cancer, inhibitory drugs, semi-tensor product, controllability, observability.



\section{Introduction}
The emergence of `systems biology' opens up new frontiers in medicine and biology. One of the major challenges in modern medical science has been to provide a cure for complex systemic diseases like cancer, Alzheimer's, and Parkinson's, to name a few. Cancer is taken as the model system for the work in this manuscript. The futility of traditional therapeutic procedures in developing one wonder-drug to cure cancer has forced physicians to look into targeted combination therapy \cite{Saito2015,sherbet2015,Reardon2015}. For the effective efforts towards the design of optimal combination therapy, it is required to have a reliable modeling scheme of the cancerous system such that a `healthy' response, `proliferating' response, and `post-therapy' response can be compared quantitatively. With such modeling scheme, a drug optimization problems can be defined which bring all these responses quantitatively closer to minimize effects of a disease.

Kauffman \cite{Kauffman1969} stated the possibility of modeling gene interactions with the Boolean network. The quantized states (ON and OFF) of the gene were represented as 0 and 1 in the logical domain. Since then, the Boolean modeling style has become famous among researchers \cite{Shmulevich2002, Shmulevich2002a, layek2011a} because of its discrete nature. The discrete nature helps to keep the computational complexity tractable for biological problems. Therefore, Boolean map (BM) and Boolean control network (BCN) have been taken as standard Boolean systems for modeling gene-regulatory networks in this manuscript.

A detailed study of the biological aspect of cancer was given by \cite{weinberg2013}. Systemic diseases like cancer begin with random mutations in somatic cells \cite{weinberg2013,martincorena2015,shendure2015} causing faults in the model and alter its desired dynamics. The altered dynamics of the system often results in an undesirable output vector (disease phenotypes). Depending on the nature of mutation, these mutations are mapped as different stuck-at faults in the Boolean systems \cite{layek2011b,abramovici1990}. These faults can be categorized into three types. If the mutated gene gets deleted or becomes incapable of transcribing and translating into the protein, the fault is of `stuck-at 0' (sa-0) type. Sometimes the mutation is such that the folded protein becomes constitutively active in the relevant signaling pathway. Then, the fault is considered to be of `stuck-at 1' (sa-1) type. The `no-fault' scenario appears when the mutation does not alter the kinase binding residues of the protein and the wild-type behavior continues \cite{layek2011b}.
The available targeted drugs for cancer are mostly small inhibitory molecules capable of blocking cell signaling in transduction pathways \cite{layek2011b}. Hence, the drugs can be modeled as {\it de-novo} inputs to the Boolean system.  These drugs control the network behavior; however, the problem of finding suitable drugs is entirely dependent on the ability to estimate the possible mutations (faults). The detail description of the mentioned biological terms is given in \cite{weinberg2013}. 

With such fault and drug modeling in the Boolean domain, a method to obtain optimal drugs was chalked out by \cite{layek2011b}, which showed a promising preliminary result towards predictive combination therapy design for single fault scenario (SSF). However, the possibility of the number of mutations is arbitrary in a general biological system. Therefore, methods are required to detect the presence of multiple stuck-at faults (MSF).  Such fault detection problems are NP-complete \cite{Fujiwara1990,cook71}. An interesting approach to fault analysis in Boolean networks was given by \cite{Fornasini2015}. However, the assumption of `fault sequence' used by them was perhaps unnecessary for real biological systems because the timescale of variation in fault vector (via random mutation) is the order of magnitude higher than the timescale of any fault diagnosis or therapy of the faulty system. Also, a trajectory-based design is not completely reliable in biological systems, as the experimental samples are usually collected non-uniformly. There is a possibility of missing certain critical transient states in such experiments. {
Li and Wang \cite{Li2012c} developed a `test set' design approach for combinational networks or Boolean maps (BMs). Liu \textit{et al}. \cite{Liu2014} extended this method for multi-valued logical maps. However, both of these papers did not consider the sequential networks or Boolean control networks (BCNs) which are considerably more complicated than a BM, especially for the multiple fault detection \cite{abramovici1990}.  Apart from this, an important class of `no-fault' (absence of a fault) was not considered in \cite{Fornasini2015,Li2012c,Liu2014}.  All these aspects have been considered for the work in this manuscript.

Unlike Boolean models, it may not always be possible to fix the specific `test patterns' in the biological network for the detection of faults (mutations). For this reason, the test set based methods like \cite{Lin2011,Li2012c,layek2012fault} may not be used in a real biological scenario. For undetectable faults, an algorithm for improvement in observability is given in this manuscript. Some earlier aspects of observability problems were considered in  \cite{Xu2013b,Laschov2013,Cheng2016,zhang2016}. In this manuscript, the optimal drugs are estimated with this fault information. If available drugs are unable to provide a suitable cure, an improvement in controllability has been suggested which provide new target locations.

With the above background, the motivation behind the underlying work resides in attempting a unified model based approach towards the realistic fault identification and intervention problems for the two classes of Boolean systems, namely the BM and the BCN. The initial results of this work were published in \cite{self2017acc}. Note that, the terms of drug (vector) and intervention are used with the same meaning interchangeably in this paper.
\section{Preliminaries}
\subsection{Variables in logical domain} \label{sec:logicdomain}
Let a logical domain be defined as 
\[
\sv{L} = \{True = 1,  False = 0\}. 
\]
Let $u$ be the input variable and $y$ be the output variable. For the logical domain,$\{u, y\} \in \sv{L}$. $d$ is a logical variable for a drug such that $d = 1$ means the drug is applied and $d = 0$ means the drug is not applied, i.e., $d \in \sv{L}$. $f$ is a ternary variable for a fault with values \{sa-1, sa-0, no-fault\} which are represented as $f \in \{-1,0,1\}$.

For $\alpha$ number of primary inputs, the input state (Boolean vector) is defined as $U = (u_1,\hdots,u_\alpha)$, $U \in \{0,1\}^\alpha$. Similarly, for $\beta$ number of primary outputs, the output state is $Y \in \{0,1\}^\beta$; for $\lambda$ number of drugs, the drugs state is $D \in \{0,1\}^\lambda$; and for $\gamma$ number of faults, the fault state is $F \in \{-1,0,1\}^\gamma$. The notations $\alpha, \beta, \gamma,$ and $\lambda$ are used in this same context throughout the manuscript.
%

\subsection{Boolean control network and Boolean map} \label{BM}
\subsubsection{Boolean control network}
A Boolean control network (BCN) \cite{Kauffman1969,Kauffm1993,Glass1973,de2002} is a discrete-time, discrete-state, deterministic system. Due to the presence of a feedback loop, it can be perceived similar to a sequential network in Boolean domain.
\begin{defn}:
A BCN is a uniformly sampled discrete state dynamical system represented as:
\begin{eqnarray} 
\begin{aligned}
X(\tau) &= \Psi^\ast \big(U(\tau),X(\tau-1)\big) , X(0) = X_0, \\
Y(\tau) &= \Gamma^\ast \big(U(\tau),X(\tau)\big),
\end{aligned} \label{eq:BCN_def}
\end{eqnarray}
where $\Psi^\ast : \{0,1\}^\alpha \times \{0,1\}^N \rightarrow \{0,1\}^N$ and $\Gamma^\ast : \{0,1\}^\alpha \times \{0,1\}^N \rightarrow \{0,1\}^\beta$ are two Boolean operators governing the dynamics of the BCN. $N$ is the number of feedback variables,  $\tau \in {\mathcal{Z}^{+}}$ is a non-negative discrete time index, and the feedback state $X(\tau) \in \{0,1\}^N$. The input state $U(\tau) \in \{0,1\}^\alpha$, and the output state $Y(\tau) \in \{0,1\}^\beta$ are the Boolean vectors as defined in Section~\ref{sec:logicdomain}.
\end{defn}
\subsubsection{Boolean map} 
A Boolean map (BM) is obtained by removing the feedback path from BCN, i.e., it behaves like a combinational network in Boolean domain. 
\begin{defn}:
A BM is a non-feedback form of a BCN which relaxes a requirement for discrete uniform sampling, represented as:
\begin{eqnarray}
\begin{aligned}
Y(t) &= \Gamma^\ast \big(U(t)\big), \label{eq:BM_def}
\end{aligned}
\end{eqnarray}
where $\Gamma^\ast : \{0,1\}^\alpha \rightarrow \{0,1\}^\beta$ is the Boolean operator. $t\in {\mathcal{R}}^{+}$ is a non-negative continuous time index.
\end{defn}

Functions ($\Psi^\ast$, $\Gamma^\ast$) for BCN and a function $\Gamma^\ast$ for BM are derived from  biological information \cite{layek2011b}. Therefore, they are assumed to be known for the work in this paper.
\subsection{Problem formulation in Boolean systems} \label{sec:prob_def}
A faulty BCN can be modeled by considering the faults as separate inputs to the system provided a fault-free BCN and the probable gene mutations (fault points). The drugs are also treated as external input variables in the known BCN topology. Without loss of generality, it is assumed that a fault-free BCN shown in (\ref{eq:BCN_def}) can be converted to define a faulty BCN as:
\begin{equation} 
\begin{aligned}
X(\tau) &= \widehat{\Psi} \big(U(\tau), F, D,X(\tau-1) \big),  X(0) = X_0,\\  
Y(\tau) &= \widehat{\Gamma} (U(\tau),X(\tau)), \label{eq:main_eq}
\end{aligned}
\end{equation}
where $\widehat{\Psi} :  \{0,1\}^\alpha \times \{-1,0,1\}^\gamma \times \{0,1\}^\lambda \times \{0,1\}^N \rightarrow \{0,1\}^N$  and $\widehat{\Gamma} :  \{0,1\}^\alpha \times \{0,1\}^N \rightarrow \{0,1\}^\beta$. Here $(\widehat{\Psi}, \widehat{\Gamma})$ is a faulty version of a fault-free BCN $(\Psi^\ast, \Gamma^\ast)$. The faults ($F$) and drugs ($D$) are assumed to be constant over time, hence, the time stamp is not mentioned. 

Similarly, a fault-free BM shown in (\ref{eq:BM_def}) can be converted to define the faulty BM as:
\begin{equation} 
\begin{aligned}
Y &= \widehat{\Gamma} (U, F, D), \label{eq:main_eq2}
\end{aligned}
\end{equation}
where $\widehat{\Gamma} :  \{0,1\}^\alpha \times \{-1,0,1\}^\gamma \times \{0,1\}^\lambda \rightarrow \{0,1\}^\beta$. $\widehat{\Gamma}$ is a modified BM after introduction of faults.  Since the time indices of the input and output vectors of the BM are continuous in nature and devoid of causality, the indices can be omitted for all purposes. Hence (\ref{eq:main_eq2}) and all the subsequent analysis of BM no longer have a time stamp.

The drugs ($D$) are inhibitory and are inherently different from the inputs ($U$). The input $U$ is mostly not under the control of the therapist for {\it in situ} modeling of a gene regulatory system. The inputs can be growth factors, hormones, oxygen, or different kind of molecular stresses \cite{weinberg2013}, whereas, drugs are user designed molecules targeted at known locations in Boolean systems. Hence the problem of control design in biological systems is finding drugs $D$. 
In both cases shown in (\ref{eq:main_eq}) and (\ref{eq:main_eq2}), $\widehat{\Psi}$ and $\widehat{\Gamma}$ are assumed to be known because $\Psi^\ast$ and $\Gamma^\ast$ are known from their construction.

\subsection{Semi-tensor products}
For the conversion of logical expression into the linear form, the variables in the logical domain are required to map in the vector form. Let a vector (delta) set $\Delta_k$ be defined as $\Delta_k = \{\delta^i_k | \,i = 1,2,\hdots,k \}$, where $\delta^i_k$ is the $i$-th column of an identity matrix ${\rm{I}}_k$. A set $\Delta_2$ is used to denote the binary values, such that $\{1,0\} \sim \{\delta_2^1, \delta_2^2\}$ respectively. Therefore, the variables defined in the logical domain are defined in the vector form given as, $\{\tv{u, y ,d}\} \in \Delta_2$.
Similarly, a vector set $\Delta_3$ is required to show a ternary variable in the vector form. Therefore, a vector form of fault variable is defined as $\tv{f} \in \Delta_3$, where $\delta_3^1$ shows a sa-1 fault, $\delta_3^2$ shows a sa-0 fault, and $\delta_3^3$ shows a no-fault condition. Although $f$ is not a binary variable, such multi-valued mapping has been done earlier in the literature \cite{Liu2014}.

The logical operators $\widehat{\Psi}$ and $\widehat{\Gamma}$ are converted into linear operators using the semi-tensor product (STP) approach formulated in  \cite{Cheng2001,Cheng2007,Cheng2009a,Cheng2010b,Cheng2011Book}.
\begin{defn}: 
Let $\sv{M}$ be the set of all matrices. Consider a matrix $\mv{A} \in \sv{M}_{m \times n}$ and a matrix $\mv{M} \in \sv{B}_{p \times q}$. Let $c$ be the least common multiple of $n$ and $p$. Then the STP of $\mv{A}$ and $\mv{B}$ is defined as
\[
\mv{A} \ltimes \mv{B} = (\mv{A} \otimes \mv{I}_{c/n})(\mv{B} \otimes \mv{I}_{c/p}),
\]
where $\otimes$ is the Kronecker product of matrices.
\end{defn}

\begin{rem}:
Every matrix product has been assumed to be an STP throughout the paper. Therefore, the notation ``$\ltimes$" is mostly excluded.
\end{rem} 

For an STP of $p$ k-valued logical variables, a mapping $\ltimes_{i=1}^p : \Delta_k \rightarrow \Delta_{k^p}$. For $\alpha$ number of input variables, an input state $U$ is represented by a vector $\tv{U} = \tv{u}_1\cdots \tv{u}_\alpha = \ltimes^\alpha_{i=1} \mathbf{u}_i$. As $\tv{u}_i \in \Delta_2$, $\tv{U} \in \Delta_{2^\alpha}$. Similarly, the output vector $\tv{Y} \in \Delta_{2^\beta}$, the fault vector $\tv{F} \in \Delta_{3^\gamma}$, and the drug vector $\tv{D} \in \Delta_{2^\lambda}$. 

\subsection{Structure matrix}
A structure matrix is derived from the network structure, and thus represents the characteristics of the network.
\begin{defn}:
For a logical function $\sigma : \sv{L}^n \rightarrow \sv{L}^m$, a matrix $\mv{M}_\sigma \in \sv{B}_{2^m \times 2^n}$ is defined as a structure matrix if
\[
\sigma(x_1,\hdots,x_n) \equiv \mv{M}_\sigma \tv{x}_1 \cdots \tv{x}_n \quad \forall x_i \in \sv{L} \text{ and } \tv{x}_i \in \Delta_2
\]
\end{defn}
The dimensions of a structure matrix depend upon the dimensions of individual $\tv{x}_i$. In the matrix form, $\mv{M}_\sigma = [\delta_{2^m}^{i_1},$ $\hdots, \delta_{2^m}^{i_{2^n}}] = \delta_{2^m}[i_1, \hdots,{i_{2^n}}]$.
\section{Methodology}
\subsection{Structure matrix for faults and drugs}
As mentioned in Section \ref{sec:prob_def}, a fault is modeled as an external input. Therefore, each fault changes the structure matrix of the network. If a node $x$ is observed (by gene sequencing) to be mutated, its value is modified to $x^\ast$.  If the observed mutation is mapped as fault $f$, then the change in the value of $x$ is given by the relation in matrix form as $\mathbf{x}^\ast = {\rm{M}}_f\,\mathbf{x\,f}$, where $\{\mathbf{x}^\ast, \mathbf{x}\} \in \Delta_2, \mathbf{f} \in \Delta_3$. In this expression, $\rm{M}_f = \delta_2[1, 2, 1, 1, 2, 2]$ is the structure matrix showing the effect of fault for different values of node $x$. For instance, if the original value of a node $x$ is 0 ($\tv{x} = \delta_2^2$) and the fault observed is sa-1 ($\tv{f} = \delta_3^1$), then $\tv{x}^\ast = \mv{M}_f \delta_2^2 \delta_3^1 = \delta_2^1$. i.e., the value of $x$ changes to 1 as an effect of fault.

Similarly, the application of an inhibitory drug $d$ blocks the node $x$ by changing its value to 0. When the drug is not applied, the node $x$ retains its value. In logical form, it is represented as $x^\ast = x\cdot \bar{d}$ \cite{layek2011b}. Using the STP techniques \cite{Cheng2010b}, this expression can be written in vector form as $\mathbf{x}^\ast = {\rm{M}}_D\, \mathbf{x\,d}$, where $\{\mathbf{x}^\ast, \mathbf{x, d}\} \in \Delta_2$ and ${\rm{M}}_D = \delta_2[2,1,2,2]$. $\mv{M}_D$ shows a structure matrix incorporating the effect of a drug applied at node $x$. The derivation of structure matrices $\mv{M}_f$ and $\mv{M}_d$ are provided in Section~\ref{supsec:intro} of the \textit{Appendix}.

The matrices $\mv{M}_f$ and $\mv{M}_D$ are constant. For every fault $f$ at a node $x$, $\tv{x}$ has to be replaced by $\mv{M}_f \tv{x f}$ to obtain a linear form. Similarly, for every inhibitory drug $d$ at a node $x$, $\tv{x}$ has to be replaced by $\mv{M}_D \tv{x d}$. 
\subsection{A linear form representation of BM and BCN}
Let $\sv{B}$ be a set of matrices with all binary elements. 
\subsubsection{Boolean map} A linear form equation for BM is obtained by applying STP properties \cite{Cheng2010b} on (\ref{eq:main_eq2}) to get a linear form representation as:
\begin{align}
\tv{Y} &= \Gamma \tv{U F D}, \label{eq:CBN2}
\end{align}
where $\Gamma \in \sv{B}_{2^\beta	\times 2^{\alpha + \lambda} 3^\gamma}$ is a structure matrix of the BM.

\subsubsection{Boolean control networks} Similar to BM, STP properties \cite{Cheng2010b} are applied on (\ref{eq:main_eq}) to get a linear form equation as:
\begin{equation}
\begin{aligned}
\tv{X}(\tau) &= \Psi \tv{U}(\tau) \tv{F D} \tv{X}(\tau-1), \\
\tv{Y}(\tau) &= \Gamma \tv{U}(\tau) \tv{X}(\tau), \label{eq:SBN2}
\end{aligned}
\end{equation}
where $\Psi \in \sv{B}_{2^N	\times 2^{\alpha + \lambda + N} 3^\gamma}$ and $\Gamma \in \sv{B}_{2^\beta	\times 2^{\alpha + N}}$ are the structure matrices of the BCN.

It is important to modify the system representation for extension applications like improvement in observability and controllability in such a way that the extensions become computationally tractable. Since the possible locations of the mutations are known in the BCN,  a BCN model is divided into blocks.  Block diagram for a BCN is given in Fig.~\ref{fig:blockwiseB}. Such block-wise division is possible with all the BCN and BM.

\begin{figure}[!ht]
\centering
\includegraphics[scale=0.59]{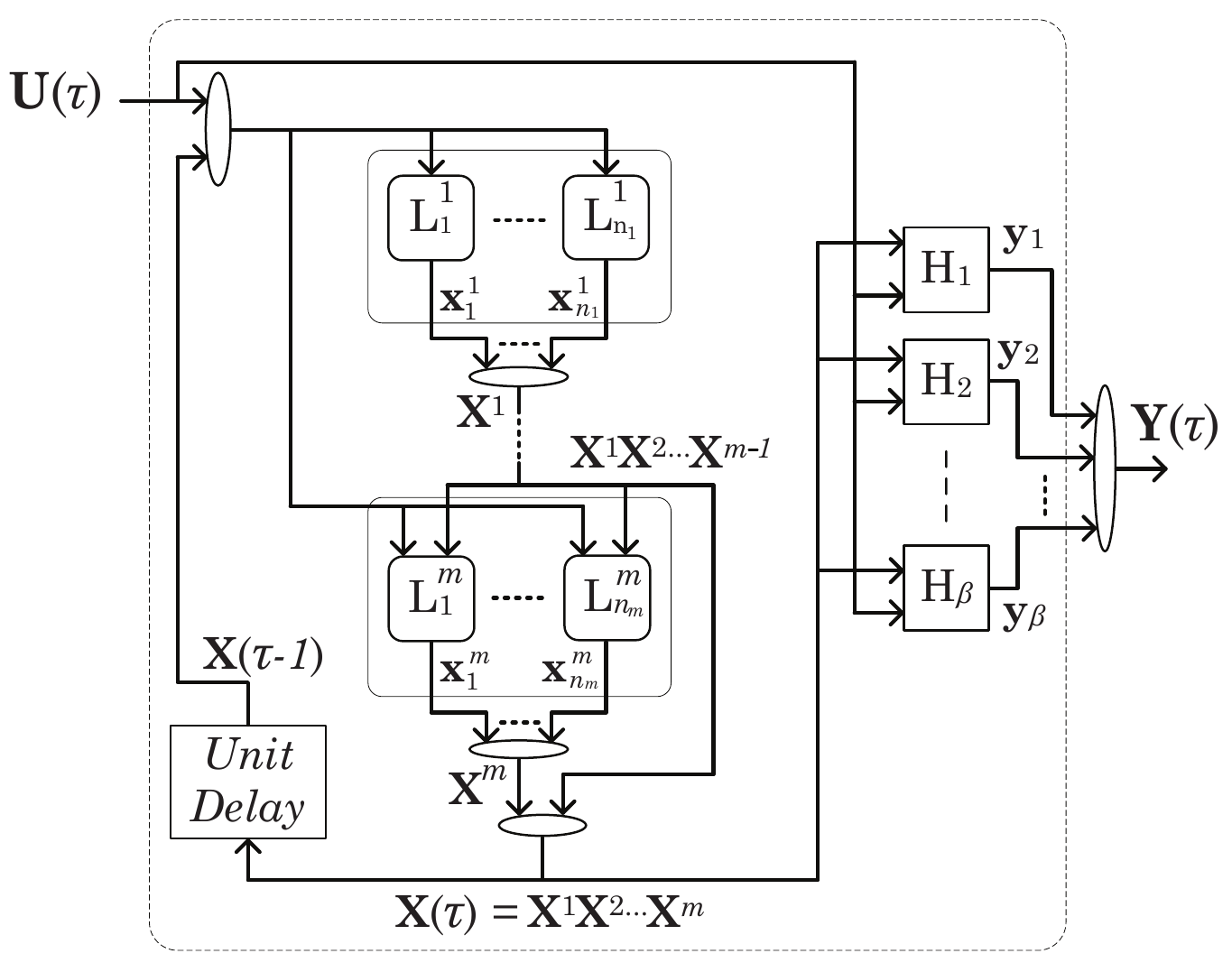} 
\caption{Blockwise design of a BCN} \label{fig:blockwiseB}
\end{figure}

The block diagram for the BCN consists of two main parts: primary block (indicated by $\rm{L}$ matrices) and the secondary block (indicated by $\rm{H}$ matrices). The primary block is further divided into sub-blocks, which are arranged level-wise based on the following rules: 

\textit{Level} 1 contains $n_1$ number of sub-blocks (with structure matrices ${\rm{L}}^1_i, i = 1,\hdots,n_1$) such that each sub-block depends on the input (including the state feedback), and only one fault or one drug. $x_i^1$ shows the output of block ${\rm{L}}_i^1$.
\textit{Level} 2 contains $n_2$ number of sub-blocks (with structure matrices ${\rm{L}}^2_i, i = 1,\hdots,n_2$) such that each sub-block depends on the input (including the state feedback), only one fault or one drug and at least one output from level 1 (i.e, at least one of $x_1^1,\hdots,x^1_{n_1}$). 
\textit{Level} $m$ contains $n_m$ number of sub-blocks (with structure matrices ${\rm{L}}^m_i, i = 1,\hdots,n_m$) such that each sub-block depends on the input (including the state feedback), only one fault or one drug and at least one output from level $(m-1)$ (i.e, at least one of $x_1^m,\hdots,x^m_{n_m}$). 

The secondary block is divided into the sub-blocks according to the following rules.
Each sub-block (${\rm{H}}_1,\hdots,{\rm{H}}_\beta$) depends upon the input and at least one output from level $m$ (i.e. $X^1,\hdots,X^m$). Each sub-block has a single output, which is one of the primary outputs, and no sub-block has any fault or drug.

The feedback input (${\rm{X}}(\tau-1)$) is taken from last level (${\rm{L}}^m$) of the sub-blocks and then it is applied to all sub-blocks of the primary block.  
The total number of sub-blocks in the network is given by: 
\begin{eqnarray}
\label{block_eqn}
{\mathcal{N}} = \sum^{m}_{i=1}{n_i} + \beta = \gamma + \lambda + \beta
\end{eqnarray}

This shows that the number of accessible nodes is limited. Since only one fault or one drug per sub-block is allowed in the block-wise separation, the generation of modified structure matrix for any choice of faults and drugs becomes straightforward. It also becomes possible to directly access the important internal nodes (either the node is a mutation site or the node is a target for an available drug). 
A structure matrix is derived for a BCN using this block-wise design.

For the BM, the construction is identical except the state feedback link. The structure matrix calculations for BM are provided in Section~\ref{supsec:BMder} the \textit{Appendix}. A similar process can be followed to obtain the structure matrices for BCN.  The final form of the linear equation of BM can be represented as:
\begin{equation}
\tv{Y} = \mv{H}\tv{UFD}, \label{eq:PO}
\end{equation}
where $\mv{H} \in \sv{B}_{2^\beta \times 2^{\alpha + \lambda} 3^\gamma}$ is a structure matrix of the BM. In comparison with (\ref{eq:CBN2}), $\Gamma = \mv{H}$ for the BM structure.

The final form of the linear equation of BCN can be expressed as:
\begin{equation}
\begin{aligned}
\tv{X}(\tau) &= \mv{L} \tv{U} \tv{F D X}(\tau-1) \\
\tv{Y}(\tau) &= \mv{H} \tv{U} \tv{F D X}(\tau), \label{eq:Struc2_FB}
\end{aligned}
\end{equation}
where $\mv{L} \in \sv{B}_{2^N	\times 2^{\alpha + \lambda + N} 3^\gamma}$ and $\mv{H} \in \sv{B}_{2^\beta	\times 2^{\alpha + N}}$ are the structure matrices of BCN. (\ref{eq:Struc2_FB}) shows that effect of faults and drugs on the next state and on the output vector. The input remains unchanged with the state transitions; therefore, the time index for the input vector is neglected.  In comparison with (\ref{eq:SBN2}), $\Psi = \mv{L}$ and $\Gamma = \mv{H W}_{[2^\lambda 3^\gamma, 2^\alpha]} \tv{FD}$ for the BCN structure. $\mv{W}_{[p,q]}$ is a swap matrix of dimensions $pq \times pq$ \cite{Cheng2010b}.
\section{Results}
The results for fault analysis and intervention for BM and BCN are derived from their structure matrices shown in (\ref{eq:PO}) and (\ref{eq:Struc2_FB}) respectively. The inputs of the biological network are assumed to be experimentally readable. Therefore, the set of input vectors is assumed to be known experimentally.  Before going to the theorems, let us introduce some variables. 
$\mathbf{F}_0$ represents a fault vector, where all the $\gamma$ faults are in `no-fault' state. Hence, $\mathbf{F}_0 = \ltimes^\gamma_{j=1} \delta_3^3 = \delta_{3^\gamma}^{3^\gamma}$. $\mathbf{D}_0$ represents a drug vector when no drug is applied. Hence, $\mathbf{D}_0 = \ltimes^\lambda_{j=1} \delta_2^2 = \delta_{2^\lambda}^{2^\lambda}$. Let ${\mathcal{F}}$ represent the set of all ($3^\gamma$) fault vectors. However in biology, the number of hazardous fault vectors is much less than $3^\gamma$. Thus, the set of hazardous faults is defined as $\widehat{\mathcal{F}} \subseteq {\mathcal{F}}$. Let ${\mathcal{U}}$ represent the set of all ($2^\alpha$) input vectors. In the input space, only some inputs are homeostatic inputs, which are achievable in real biological systems.  Thus, $\widehat{\mathcal{U}} \subseteq {\mathcal{U}}$ represents the set of such permissible inputs. 
Let ${\mathcal{D}}$ be the set of all ($2^\lambda$) drugs. 
%

The next section discusses the main results of fault estimation and intervention along with some corollaries and proposition derived. If the available set of inputs fails to identify a particular fault uniquely, then it requires designing new output (reporters) to improve the observability.  Similarly, if the drugs are not available for the estimated faults, the improvement in controllability is required. Algorithms 1 and 2 show such a possibility of improvement in observability and controllability of the network. The results are initially developed for the Boolean map and then extended for the Boolean control networks. 
\begin{rem}:
Proofs of the two main theorems are given. All the other proofs are omitted as those proofs can be given following the same arguments. 
\end{rem}

\subsection{Boolean map} \label{BM2}
The network information of BM is available in the structure matrix $\mv{H}$ represented by (\ref{eq:PO}). Therefore, the fault analysis and intervention results are derived from the structure matrix $\mv{H}$. 
\begin{thm}:  \label{theorem:1.1}
\textit{Existence theorem}\\
For a given input vector $\widetilde{\mathbf{U}} \in \widehat{\mathcal{U}}$ in a BM in (\ref{eq:PO}),  the existence of fault vector $\mathbf{F}_i$ is assured iff
\begin{equation*}
\widetilde{\rm{H}}\mathbf{\widetilde{U}F}_i \ne \widetilde{\rm{H}}\mathbf{\widetilde{U}F}_0,
\end{equation*} 
where $\widetilde{\rm{H}} = {\rm{HW}}_{[2^\lambda,2^\alpha 3^\gamma]}\mathbf{D}_0$ and $\mv{W}_{[p,q]}$ is a swap matrix of dimensions $pq \times pq$ \cite{Cheng2010b}.\\
\end{thm}
\begin{proof}:
{\it Necessary condition:} For the specified input vector $\ttv{U} \in \whsv{U}$, a null-drug vector $\mathbf{D}_0$, and an arbitrary fault vector $\mathbf{F}_i$ in a BM, output $\mathbf{Y}(\widetilde{\mathbf{U}},\mathbf{F}_i,\mathbf{D}_0)$ of the system is given by: 
\begin{eqnarray*}
\mathbf{Y}(\widetilde{\mathbf{U}},\mathbf{F}_i,\mathbf{D}_0) = {\rm{H}}\widetilde{\mathbf{U}}\mathbf{F}_i{\mathbf{D}_0}
			  = {\rm{HW}}_{[2^\lambda,2^\alpha 3^\gamma]}\mathbf{D}_0 \widetilde{\mathbf{U}}\mathbf{F}_i 
              = \widetilde{\rm{H}}\widetilde{\mathbf{U}}\mathbf{F}_i,  
\end{eqnarray*}
where $\widetilde{\rm{H}} = {\rm{H W}}_{[2^\lambda,2^\alpha 3^\gamma]} \mathbf{D}_0$. $\tmv{H} \in \sv{B}_{2^\beta \times 2^\alpha 3^\gamma}$.  If a system with structure matrix $\rm{H}$ is fault-free or under the influence of null-fault vector $\mathbf{F}_0$, the output is given by: 
\begin{eqnarray*}
\mathbf{Y}(\widetilde{\mathbf{U}},\mathbf{F}_0,\mathbf{D}_0) = {\rm{H}}\widetilde{\mathbf{U}}\mathbf{F}_0{\mathbf{D}_0}
			  = {\rm{HW}}_{[2^\lambda,2^\alpha 3^\gamma]}\mathbf{D}_0 \widetilde{\mathbf{U}}\mathbf{F}_0 
              = \widetilde{\rm{H}}\widetilde{\mathbf{U}}\mathbf{F}_0, 
\end{eqnarray*}
where all input conditions are kept identical. 
If $\mathbf{Y}(\widetilde{\mathbf{U}},\mathbf{F}_i,\mathbf{D}_0)$ = $\mathbf{Y}(\widetilde{\mathbf{U}},\mathbf{F}_0,\mathbf{D}_0)$, the fault vectors $\mathbf{F}_i$ and $\mathbf{F}_0$ become indistinguishable. Hence, the necessary condition of the existence of $\mathbf{F}_i$ is $\mathbf{Y}(\widetilde{\mathbf{U}},\mathbf{F}_i,\mathbf{D}_0) \neq \mathbf{Y}(\widetilde{\mathbf{U}},\mathbf{F}_0,\mathbf{D}_0)$. 
Therefore, 
\begin{align*}
\widetilde{\rm{H}}\widetilde{\mathbf{U}}\mathbf{F}_i &\ne \widetilde{\rm{H}}\widetilde{\mathbf{U}}\mathbf{F}_0.
\end{align*}
\\{\it Sufficient condition:}
Suppose, 
\begin{align*}
\widetilde{\rm{H}}\widetilde{\mathbf{U}}\mathbf{F}_i &\ne \widetilde{\rm{H}}\widetilde{\mathbf{U}}\mathbf{F}_0,
\end{align*}
where  $\widetilde{\rm{H}} = {\rm{H W}}_{[2^\lambda,2^\alpha 3^\gamma]} \mathbf{D}_0$. $\widetilde{\rm{H}} $ is the modified structure matrix of the BM under null-drug vector $\mathbf{D}_0$. The above inequality trivially establishes that $\mathbf{F}_i \neq \mathbf{F}_0$. Hence it is sufficient to say that the inequality guaranties existence of fault vector $\mathbf{F}_i$. This completes the proof. 
\end{proof}
\begin{cor}:  \label{cor:thm1.0}
For a given fault vector $\widetilde{\mathbf{F}} \in \widehat{\mathcal{F}}$ in a BM in (\ref{eq:PO}),  there exists an input vector $\mathbf{U}_j \in \widehat{\mathcal{U}}$ capable of detecting a fault vector $\widetilde{\mathbf{F}}$ iff 
\begin{equation*}
\widetilde{\rm{H}}^{'}\widetilde{\mathbf{F}}\mathbf{U}_j \neq \widetilde{\rm{H}}^{'}{\mathbf{F}_0}\mathbf{U}_j,
\end{equation*}
where $\widetilde{\rm{H}}^{'} = {\rm{HW}}_{[3^\gamma 2^\lambda,2^\alpha]}{\rm{W}}_{[2^\lambda,3^\gamma]}\mathbf{D}_0$.
\end{cor}
\begin{prop}:  \label{cor:thm1.1}
For a specified input vector $\widetilde{\mathbf{U}} \in \widehat{\mathcal{U}}$ in BM in (\ref{eq:PO}), a set of fault vectors ${\mathcal{F}_{\widetilde{\mathbf{U}}}}  \subseteq \widehat{\mathcal{F}}$ which is detectable by input vector $\widetilde{\mathbf{U}}$ is given by  
\begin{align*}
{\mathcal{F}_{\widetilde{\mathbf{U}}}} = \{\mathbf{F}_i\,\mid\, \widetilde{\rm{H}}\widetilde{\mathbf{U}} \mathbf{F}_i \ne \widetilde{\rm{H}}\widetilde{\mathbf{U}} {\mathbf{F}_0} \, ;  \mathbf{F}_i \in \widehat{\mathcal{F}}  \},
\end{align*}
where $\widetilde{\rm{H}} = {\rm{H W}}_{[2^\lambda,2^\alpha 3^\gamma]} \mathbf{D}_0$. For \textit{existence}, ${\mathcal{F}_{\widetilde{\mathbf{U}}}} \ne \emptyset$.
\end{prop}
The following results discuss the necessary and sufficient conditions for detecting the existence of a particular fault using an arbitrary input.

\begin{prop}:  \label{cor:thm1.2}
For a specified fault vector $\widetilde{\mathbf{F}} \in \widehat{\mathcal{F}}$ in a BM in (\ref{eq:PO}), a set of input vectors ${\mathcal{U}_{\widetilde{\mathbf{F}}}} \subseteq \widehat{\mathcal{U}} $  which can detect the fault vector is given by:
\begin{equation*}
{\mathcal{U}_{\widetilde{\mathbf{F}}}} = \{\mathbf{U}_i\,\mid\,\widetilde{\rm{H}}^{'}\widetilde{\mathbf{F}}\mathbf{U}_i \neq \widetilde{\rm{H}}^{'}\mathbf{F}_0 \mathbf{U}_i
\, ;  \mathbf{U}_i \in \widehat{\mathcal{U}}  \},
\end{equation*}
where $\widetilde{\rm{H}}^{'} ={\rm{HW}}_{[3^\gamma 2^\lambda,2^\alpha]}{\rm{W}}_{[2^\lambda,3^\gamma]}\mathbf{D}_0$.
\end{prop}

\begin{thm}:  \label{theorem:2}
\textit{Uniqueness theorem} \\
For a specific fault $\widetilde{\mathbf{F}} \in \widehat{\mathcal{F}}$ in a BM in (\ref{eq:PO}), input $\widetilde{\mathbf{U}} \in \widehat{\mathcal{U}}$ can uniquely identify fault vector $\widetilde{\mathbf{F}}$ iff 
\begin{align*}
\widetilde{\rm{H}} \widetilde{\mathbf{U}} \widetilde{\mathbf{F}} &\ne \widetilde{\rm{H}} \widetilde{\mathbf{U}} \mathbf{F}_0, \\
\text{ and } \widetilde{\rm{H}} \widetilde{\mathbf{U}} \mathbf{F}_j &= \widetilde{\rm{H}} \widetilde{\mathbf{U}} \mathbf{F}_0; \quad\forall \mathbf{F}_j \ne \widetilde{\mathbf{F}}, \mathbf{F}_j\in \widehat{\mathcal{F}}.
\end{align*}
\end{thm}

\begin{prop}:  \label{cor:thm2}
For a specified fault vector $\widetilde{\mathbf{F}} \in \widehat{\mathcal{F}}$ in a BM in (\ref{eq:PO}), a set of input vectors ${\mathcal{U}_{\widetilde{\mathbf{F}}}} \subseteq \widehat{\mathcal{U}}$ which can uniquely identify a fault vector $\widetilde{\mathbf{F}}$ is given by:  
\begin{align*}
{\mathcal{U}_{\widetilde{\mathbf{F}}}} =\{\mathbf{U}_j|\mathbf{U}_j\in \widehat{\mathcal{U}}, \widetilde{\rm{H}} \mathbf{U}_j \widetilde{\mathbf{F}} &\ne \widetilde{\rm{H}} \mathbf{U}_j \mathbf{F}_0,\\
\widetilde{\rm{H}} \mathbf{U}_j \mathbf{F}_k &= \widetilde{\rm{H}} \mathbf{U}_j \mathbf{F}_0\quad\forall \mathbf{F}_k \ne \widetilde{\mathbf{F}}, \mathbf{F}_k \in \widehat{\mathcal{F}} \}.
\end{align*}
\end{prop}

\begin{prop}:  \label{def:com_fault}
For a given set of input vectors $\sv{U}^\ast \subseteq \widehat{\mathcal{U}}$, any input vector $\mathbf{U}_i \in {\mathcal{U^{*}}}$ can detect a set of fault vectors ${\mathcal{F}_{{\mathbf{U}}_i}}$ (\text{Proposition}~\ref{cor:thm1.1}). The set of common fault vectors detectable by a set of given input vectors  can be estimated as:
\begin{equation*}
{\mathcal{F}_{{\mathcal{U}}^\ast}} = \bigcap_{i=1}^{\xi}{\mathcal{F}_{{\mathbf{U}}_i}}\,;\,\mathbf{U}_i \in {\mathcal{U}}^\ast,
\end{equation*}
where $\xi = card({\mathcal{U}}^\ast)$ is the cardinality of the set ${\mathcal{U}}^\ast$.
\end{prop}

\begin{prop}: \label{cor:test_set}
Let $\widehat{\mathcal{F}}$ be a set of permissible fault vectors for a BM in (\ref{eq:PO}) and a fault vector $\mathbf{F}_i \in \widehat{\mathcal{F}}$. Let ${\mathcal{U}_{{\mathbf{F}}_i}}$ be a set of input vectors  that can detect a fault vector $\mathbf{F}_i$ (see {\text{Proposition}~\ref{cor:thm1.2}}). Then, an input test set ${\mathcal{U}}_\text{T}$ can be generated as:
\begin{equation*}
{\mathcal{U}}_{\rm{T}}  = \bigcup_{i=1}^{\mu} {\mathcal{U}_{{\mathbf{F}}_i}} \, ; \, \mathbf{F}_i \in \widehat{\mathcal{F}},
\end{equation*}
where $\mu = card(\widehat{\mathcal{F}})$.
\end{prop}
The test set generated by this method is not optimal, but it is useful to know the important homeostatic inputs and the fault coverage.

\begin{prop}:
\label{fault_coverage}
For a test set ${\mathcal{U}}_\text{T}$ for a BM in (\ref{eq:PO}), fault coverage ${\mathcal{F}}_{C} \subseteq \widehat{\mathcal{F}}$ is given by
\begin{equation*}
{\mathcal{F}}_{C} = \bigcup_{i = 1}^\zeta {\mathcal{F}_{{\mathbf{U}}_i}},
\end{equation*}
where $\zeta = card ({\mathcal{U}}_\text{T})$ and ${\mathcal{F}_{{\mathbf{U}}_i}}$ is set of faults detectable by input vector $\mathbf{U}_i \in {\mathcal{U}}_\text{T}$ (see {\text{Proposition}~\ref{cor:thm1.1}}).
\end{prop}
Due to the limited number of homeostatic inputs, many of the fault vectors may remain undetectable in a real biological system. This affects the fault coverage. The following theorem provides a way to use multiple input vectors for unique fault detection. 

\begin{cor}:  \label{theorem:3}
\textit{Generalized uniqueness} \\
A set of input vectors ${\mathcal{U}}^\ast \subseteq \widehat{\mathcal{U}}$ can uniquely determine a fault vector $\widetilde{\mathbf{F}} \in {\mathcal{F}}$ iff ${\mathcal{F}_{{\mathcal{U}}^\ast}}  = \{\widetilde{\mathbf{F}}\}$. 
\end{cor}

\begin{prop}:
If ${\mathcal{F}_{{\mathcal{U}}^\ast}}  = \emptyset$, a set of input vectors ${\mathcal{U}}^\ast$ cannot detect any fault vector. 
\end{prop}
\begin{prop}:
If $card({\mathcal{F}_{{\mathcal{U}}^\ast}}) \neq 1 $, a set of input vectors ${\mathcal{U}}^\ast$ cannot detect any fault vector uniquely. 
\end{prop}

\begin{defn}:
If $card({\mathcal{F}_{{\mathcal{U}}^\ast}}) > 1$, a members of the set ${\mathcal{F}_{{\mathcal{U}}^\ast}} $ are called \textit{indistinguishable} faults under a set of input vectors $\mathcal{U^{*}}$.
\end{defn}
%
\begin{defn}:
If two fault vectors $\mathbf{F}_j$ and $\mathbf{F}_k$ are indistinguishable and 
\begin{equation*}
\widetilde{\rm{H}}\mathbf{U}_j\mathbf{F}_i = \widetilde{\rm{H}}\mathbf{U}_j \mathbf{F}_k\, ;\, \forall \mathbf{U}_j \in \mathcal{U^{*}}\,, \,  \mathbf{F}_i, \mathbf{F}_k \in {\mathcal{F}_{{\mathcal{U}}^\ast}}\,, \, \mathbf{F}_i \ne \mathbf{F}_k,
\end{equation*}
then fault vectors $\mathbf{F}_i$ and $\mathbf{F}_k$ are defined as \textit{equivalent}. Symbolically $\mathbf{F}_i \equiv \mathbf{F}_k$.
\end{defn}
%
In {Proposition~\ref{cor:thm1.2}}, if some fault vector results in ${\mathcal{U}_{\mathbf{F}}} = \emptyset$ then that fault vector is said to be undetectable. Such faults may not be harmful in the pathways under consideration, but may be harmful in other dependent pathways. Therefore it is important to detect such faults. 
 In {Corollary~\ref{theorem:3}}, if set ${\mathcal{F}_{{\mathcal{U}}^\ast}}$ represents equivalent fault vectors, then fault detection becomes ambiguous or undetectable. For these conditions, a new reporter design is necessary. Some faults may not be detectable at primary outputs, but may be detectable at the output of internal blocks. Therefore, observability of faults at every internal node is used to decide a new reporter. 

Let a set of undetectable fault  vectors for a set of input vectors $\mathcal{U}$ be denoted as $\overline{\mathcal{F}}_\mathcal{U}$. Set $\overline{\mathcal{F}}_\mathcal{U}$ contains all the faults for which ${\mathcal{U}_{\mathbf{F}}} = \emptyset$ and the faults which are equivalent.  Algorithm~\ref{alg:observe} results into best possible reporter(s) for improved fault detection. ${\mathcal{F}^{ij}_\mathcal{U}}$ indicates a set of detectable faults considering node $x^i_j$ as `reporter'. Reporter is an additional input of the network which provides additional information for fault detection. Selection of each reporter augments the output vector by one. Algorithm~\ref{alg:observe} results in the set of best possible reporters, ${\mathcal{X}}$ (best reporter at ${\mathcal{X}}(1)$). It is clear from Algorithm~\ref{alg:observe} that the complexity of this algorithm is $O(2^{\sum^{m}_{i=1}{n_i}}) $ or $O(2^{\gamma+\lambda})$ (using (\ref{block_eqn})). The modular design keeps the structure matrices for the individual blocks accessible for algebraic manipulation.

\begin{algorithm}[!ht]
\caption{Improvement in observability} \label{alg:observe}
Define $i:= 1,\hdots,m$; $j:= 1,\hdots,n_i$; $k = 0$ ; ${\mathcal{X}} = \emptyset$ \; 
Augment $x^i_j$ as the $(\beta + 1)^\text{th}$ output \;
Compute a set of detectable faults ${\mathcal{F}}^{ij}_\mathcal{U} \subseteq \overline{\mathcal{F}}_\mathcal{U} \, \forall i,j$ ({\text{Proposition}~ \ref{def:com_fault}}) \;
\While{$\overline{\mathcal{F}}_\mathcal{U} \neq \emptyset$ or  $card(\overline{\mathcal{F}}_\mathcal{U})$ decreases}
{
$({i_p},{j_p}) = \underset{\forall i,j}{\mathrm{argmax}}$  $card ({\mathcal{F}}^{ij}_\mathcal{U})$ \;
Select $x^{{i_p}}_{{j_p}}$ as reporter \;
$\overline{\mathcal{F}}_\mathcal{U} \leftarrow \overline{\mathcal{F}}_\mathcal{U} \setminus {\mathcal{F}}^{{i_p}{j_p}}_\mathcal{U}$ \;
${\mathcal{F}}^{ij}_\mathcal{U} \leftarrow {\mathcal{F}}^{ij}_\mathcal{U} \subseteq \overline{\mathcal{F}}_\mathcal{U}$ \;
${\mathcal{X}}(k) \leftarrow {x^{i_p}_{j_p}}$ \;
$k \leftarrow k+1$ \;
}
\Return{${\mathcal{X}}$} \;
\end{algorithm}

It is possible in principle to put some additional constraints depending on some prior biological knowledge of certain nodes in choosing the best reporter in each iteration. The accessibility of the new reporters is restricted to the output nodes of the primary sub-blocks only. The number of such sub-blocks is $\gamma+\lambda$, which is small compared to the total number of internal nodes in the system keeping the computational complexity of the algorithm manageable.
\begin{thm}:  \label{theorem:5}
\textit{Existence of Intervention} \\
For a BM in (\ref{eq:PO}) with a fault vector $\widetilde{\mathbf{F}} \in \widehat{\mathcal{F}}$ and an input vector $\widetilde{\mathbf{U}} \in \widehat{\mathcal{U}}$, a drug vector $\mathbf{D}_i \in \mathcal{D}$ exists if
\begin{equation*}
\widetilde{\rm{H}} \widetilde{\mathbf{F}} \widetilde{\mathbf{U}} \mathbf{D}_i = \widetilde{\rm{H}} \mathbf{F}_0 \widetilde{\mathbf{U}} \mathbf{D}_0, 
\end{equation*}
where $\widetilde{\rm{H}} = {\rm{H W}}_{[3^\gamma, 2^\alpha]}$, and $\lambda$ is a number of available drugs. 
\end{thm}

\begin{prop}:  \label{cor1:thm5}
For a fault vector $\widetilde{\mathbf{F}} \in \mathcal{F}$ and a set of input vectors $\widehat{\mathcal{U}}$ in BM in (\ref{eq:PO}), a set of drug vector ${\mathcal{D}_{\widetilde{\mathbf{F}}}}$ which can control the network is given by:
\begin{equation*}
{\mathcal{D}_{\widetilde{\mathbf{F}}}} = \{ \mathbf{D}_j\,\mid\, \widetilde{\rm{H}}\widetilde{\mathbf{F}} \widetilde{\mathbf{U}} \mathbf{D}_j = \widetilde{\rm{H}} \mathbf{F}_0 \widetilde{\mathbf{U}} \mathbf{D}_0\, ; \mathbf{D}_j \in \mathcal{D} \text{ and } \widetilde{\mathbf{U}} \in \widehat{\mathcal{U}} \}.
\end{equation*}
\end{prop}
%
%
\begin{exmp}:  \label{example:drug}
{\rm{
Consider a BM as shown in Fig.~\ref{example:drug_intervention}.
\begin{figure}[!ht]
\centering
\includegraphics[scale=0.42]{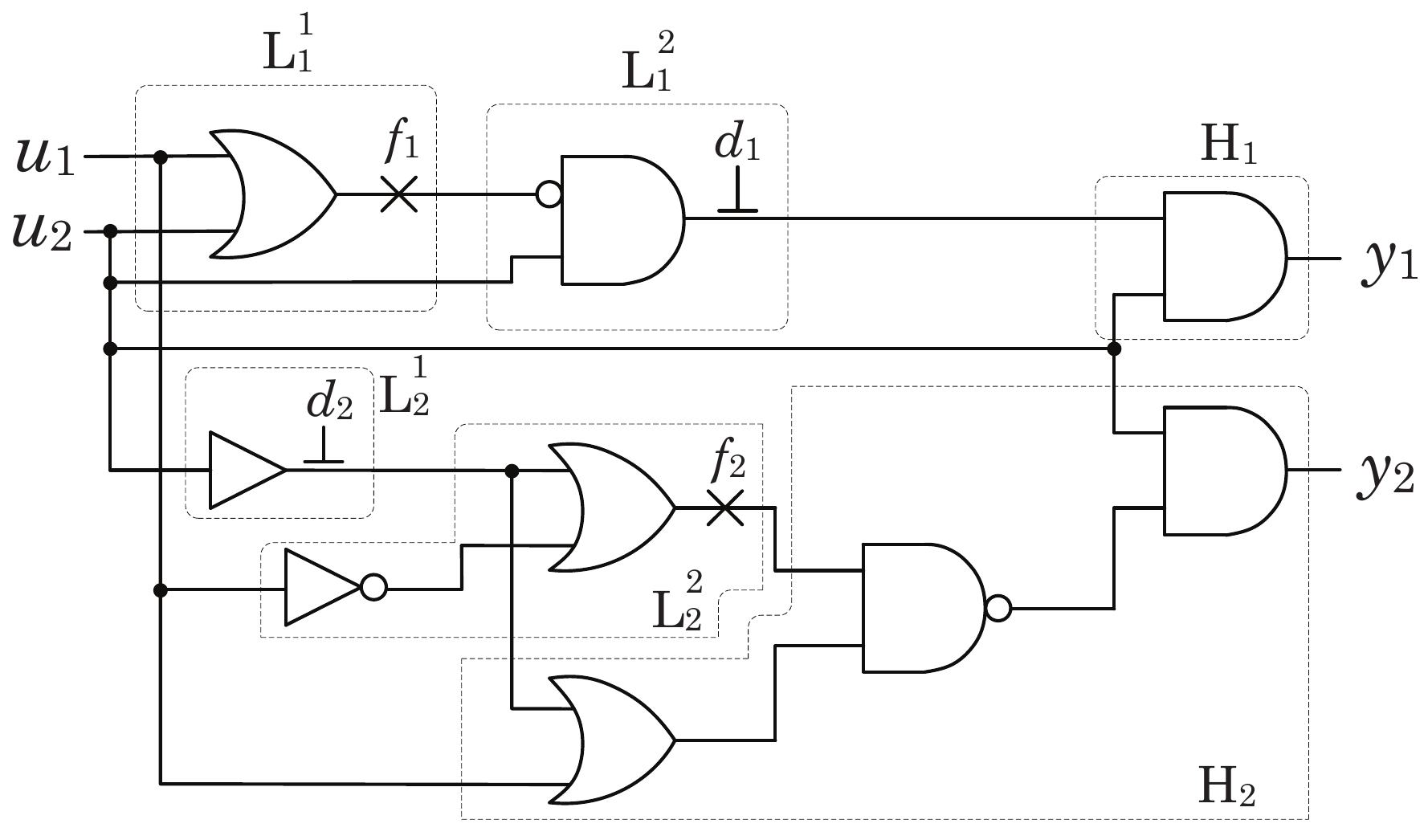}
\caption{Example 1: Boolean map} \label{example:drug_intervention}
\end{figure}
From the network, $\alpha = 2$, $\beta = 2$, $\gamma = 2$, $\lambda = 2$, $m = 2$, $n_1 = 2$, and $n_2 = 2$. For $\mathbf{Y} = \widetilde{\rm{H}}\mathbf{FUD}$, structure matrix $\widetilde{\rm{H}} = \rm{HW}_{[9,4]}$ for the system using (\ref{eq:PO}) can be easily calculated as: 
\begin{align*}
\widetilde{\rm{H}} &= \delta_4\big[4,4,4,4,4,4,4,4,3,4,3,4,4,4,4,4 \mid 3,3,3,3,4,\\
&\hspace{0cm} 4,4,4,3,3,3,3,4,4,4,4 \mid 3,4,3,4,4,4,4,4,3,4,3,4,4, \\
&\hspace{0cm} 4,4,4 \mid 4,4,2,2,4,4,4,4,3,4,1,2,4,4,4,4 \mid 3,3,1,1,4,\\
&\hspace{0cm} 4,4,4,3,3,1,1,4,4,4,4 \mid 3,4,1,2,4,4,4,4,3,4,1,2,4,\\
&\hspace{0cm} 4,4,4 \mid 4,4,4,4,4,4,4,4,3,4,3,4,4,4,4,4 \mid 3,3,3,3,4, \\
&\hspace{0cm} 4,4,4,3,3,3,3,4,4,4,4 \mid 3,4,3,4,4,4,4,4,3,4,3,4,4, \\
& 4,4,4\big]
\end{align*}
If no drugs are applied (i.e., $\mathbf{D}_0 =\delta_4^4$), the system equation can be written as $\mathbf{Y} = \widehat{\rm{H}}\mathbf{FU}$, where
 \begin{align*}
\widehat{\rm{H}} &= \delta_4\big[4,4,4,4 \mid 3, 4, 3, 4 \mid 4,4,4,4 \mid 2,4,2,4 \mid \\
&\hspace{-0.4cm} 1,4,1,4 \mid 2, 4, 2, 4 \mid 4,4,4,4 \mid 3,4,3,4 \mid 4,4,4,4\big]
\end{align*}
Let us assume an input vector $\widetilde{\mathbf{U}} = \delta_4^3$. From Proposition~{\ref{cor:thm1.1}}, a set of detectable faults vectors can be obtained as ${\mathcal{F}_{\widetilde{\mathbf{U}}}} = \delta_9\{2,4,5,6,8\}$. Similarly, for a fault vector $\widetilde{\mathbf{F}} = \delta_9^4$, using Proposition~{\ref{cor:thm1.2}}, a set of input vectors which can detect the fault vector $\widetilde{\mathbf{F}}$ is given by ${\mathcal{U}_{\widetilde{\mathbf{F}}}} = \delta_4\{1,3\}$. Suppose a set of fault vectors $\widehat{\mathcal{F}} = \delta_9\{1,2,3,7\}$ is known. Using Theorem~{\ref{theorem:2}}, an input vector $\widetilde{\mathbf{U}} = \delta_4^3$ can uniquely identify fault $\widetilde{\mathbf{F}}=\delta_9^2$. For the set of given input vectors $\widehat{\mathcal{U}} = \delta_4\{1,2,3\}$, the fault coverage (Proposition~ {\ref{fault_coverage}}) is ${\mathcal{F}}_{C} = \delta_9\{2,4,5,6,8\}$. However, for the set of input vectors $\widehat{\mathcal{U}} = \delta_4\{2,4\}$, fault coverage is ${\mathcal{F}}_{C} = \emptyset$. If a given set of fault vectors is $\widehat{\mathcal{F}} = \delta_9\{1,2,3,4,8\}$, Corollary~{\ref{theorem:3}} says that using $\widetilde{\mathbf{U}} = \delta_4^3$, there is no uniquely detectable fault vectors. In fact all the detectable fault vectors, i.e, ${\mathcal{F}_{\widetilde{\mathbf{U}}}} = \delta_9\{2,4,8\}$ are {\it indistinguishable} and $\delta_9^2 \equiv \delta_9^8$.  \\
\textit{Improvement in observability}: It is evident from the matrix $\widetilde{\rm{H}}$ that fault vectors in the set $\overline{\mathcal{F}}_{\mathcal{U}} = \delta_9\{1,3,7\}$ are undetectable for control inputs $\widehat{\mathcal{U}}$. It is required to improve the observability of the network to detect these faults.
Therefore from Algorithm 1, a reporter can be designed as follows: 
\begin{align*}
	\text{At node $x_1^1$ : } {\mathcal{F}_{\mathbf{U}_4}^{11}} &= \delta_9\{1,3\}, \\
	\text{at node $x_2^2$ : } {\mathcal{F}_{\mathbf{U}_2}^{22}} &= \delta_9\{1,7\}, 
\end{align*}
where $\mathbf{U}_2 = \delta^2_4$ and $\mathbf{U}_4 = \delta^4_4$. Therefore, node $x^1_1$ can be selected as primary reporter. Now, the set of undetectable faults is reduced to 
$\overline{\mathcal{F}}_{\mathcal{U}} = \delta_9^7$. Selecting node $x^2_2$ as secondary reporter guarantees detection of all faults. However, for constrained input set $\widehat{\mathcal{U}} = \delta_4\{1,3\}$, observability cannot be improved.\\ 
 \textit{Controllable faults:} The faults for which drugs are available are called controllable faults. Assume a set of input vectors as $\widehat{\mathcal{U}} = \delta_4\{1,3\}$. Assume the detected faults using fault estimation method as  $\mathbf{f}_1 = \delta_3^2$ (`stuck-at 0') and $\mathbf{f}_2 = \delta_3^3$ (no fault), i.e., the fault vector $\widetilde{\mathbf{F}} = \delta_9^6$. For an input vector $\widetilde{\mathbf{U}} = \delta_4^1$ or $\widetilde{\mathbf{U}} = \delta_4^3$:
\begin{align*}
\widetilde{\rm{H}}\widetilde{\mathbf{F}}\widetilde{\mathbf{U}}  &= \delta_4[3,4,1,2].
\end{align*}
For a fault-free network without drugs:
\begin{align*}
\widetilde{\rm{H}}\mathbf{F}_0\widetilde{\mathbf{U}}\mathbf{D}_0 &= \widetilde{\rm{H}}\delta_9^9 \widetilde{\mathbf{U}}\delta_4^4 = \delta_4^4. \\
\text{For $\mathbf{D} = \delta_4^2$ : }
\widetilde{\rm{H}}\widetilde{\mathbf{F}}\widetilde{\mathbf{U}}\mathbf{D} &= \widetilde{\rm{H}}\delta_9^9 \widetilde{\mathbf{U}}\delta_4^4.
\end{align*}
i.e., ${\mathcal{D}_{\widetilde{\mathbf{F}}}} = \{\delta_4^2\}$, or $\mathbf{d}_1 = \delta_2^1$ and $\mathbf{d}_2  = \delta_2^2$. This shows that for the given fault scenario $\widetilde{\mathbf{F}}=\delta_9^6$, only drug $\mathbf{d}_1$ is sufficient for nullifying the effect of the fault vector.\\
\textit{Uncontrollable faults:} The faults for which drugs are not available are called uncontrollable faults. Let detected faults using fault estimation method be  $\mathbf{f}_1 = \delta_3^2$ (`stuck-at 0') and $\mathbf{f}_2 = \delta_3^2$ (`stuck-at 0'). i.e.fault vector $\widetilde{\mathbf{F}} = \delta_9^5$. For input vector $\widetilde{\mathbf{U}} = \delta_4^1$ or $\widetilde{\mathbf{U}} = \delta_4^3$:
\begin{align*}
\widetilde{\rm{H}}\widetilde{\mathbf{F}}\widetilde{\mathbf{U}} &= \delta_4[3,3,1,1] \\
\widetilde{\rm{H}}\delta_9^9 \widetilde{\mathbf{U}}\delta_4^4 &= \delta_4^4 \\
\text{For any $\mathbf{D}$ : }
\widetilde{\rm{H}}\widetilde{\mathbf{F}}\widetilde{\mathbf{U}}\mathbf{D} &\neq \widetilde{\rm{H}}\delta_9^9 \widetilde{\mathbf{U}}\delta_4^4 
\end{align*}
Therefore the effects of the fault cannot be controlled with the available set of drugs.}} \exclose
\end{exmp}
%
\textit{Improvement in controllability}:  For Proposition~\ref{cor1:thm5}, if ${\mathcal{D}_{\widetilde{\mathbf{F}}}} = \emptyset$, then a new drug is required such that it can eliminate the effects of the faults. Let $\mathbf{d}_{\lambda+1}$ be a new drug. Then, 
\begin{equation*}
\hat{\mathbf{D}} =  \mathbf{d}_1\,\mathbf{d}_2\cdots \mathbf{d}_{\lambda+1} = {\mathbf{D\,d}_{\lambda+1}}= \ltimes^{\lambda+1}_{j=1} \mathbf{d}_j
\end{equation*}
Naturally, the effect of faults will be maximum at the downstream protein(s) in the pathways. Although the secondary block does not have any fault or drug, a new drug can be targetted at the primary output of that block. These targets facilitate better possibilities of the drug intervention. Therefore, the best possible target site for the drug is searched from output towards input, level wise. Depending upon the target location of the drug, there are two possibilities. (i) Target point is internal node ($x^i_j$ from primary block) of a BM, and (ii) target point is one of the primary outputs (secondary block) of a BM.  For these two cases, similar to (\ref{eq:PO}), equation of the output can be written as 
\begin{equation}
\mathbf{Y} = \widehat{\rm{H}}\mathbf{U F}\hat{\mathbf{D}}.
\end{equation}  
Details of the derivation process are provided in the \textit{Appendix} (Section~\ref{supsec:control}). Algorithm~\ref{alg:control} shows the  process of improvement in controllability. However, all faults cannot be controlled with inhibitor type drugs. Some activator drugs are also required for improved control.
\begin{algorithm} [!ht]
\caption{Improvement in controllability} \label{alg:control}
\While{${\mathcal{D}_{\widetilde{\mathbf{F}}}} = \emptyset$}
{
Derive $Y$ for chosen case \;
Find control set ${\mathcal{D}_{\widetilde{\mathbf{F}}}}$ using {\text{Proposition}~\ref{cor1:thm5}} \;
\If{${\mathcal{D}_{\widetilde{\mathbf{F}}}} = \emptyset$} 
{
Change target point \;
Repeat all steps \;
}
\If{${\mathcal{D}_{\widetilde{\mathbf{F}}}} \neq \emptyset$ or all target points considered}
{
Stop \;
}
}
\Return{${\mathcal{D}_{\widetilde{\mathbf{F}}}}$} \;
\end{algorithm}

\subsection{Boolean control network}  \label{BCN}
In (\ref{eq:Struc2_FB}), the information about network dynamics is shown by structure matrix $\mv{L}$. As the secondary block is fault-free, a structure matrix $\mv{L}$ alone is sufficient for the fault analysis. The cyclic attractors in the state transitions of the BCN can be calculated from the diagonal of its structure matrix \cite{Cheng2010b}.  The number and sizes of the attractor cycles in the network change with the fault present. For each input vector and fault vector, a structure matrix ${\rm{L}}_\mathbf{UF}$ and its corresponding cycles are estimated. Assume that there are $s$ different attractor cycles in the network. The cycles are numbered in the increasing size of their lengths. The length of an $i^{th}$ cycle is $l_i$. For a cycle of length $k$, the corresponding diagonal elements of the structure matrix $({\rm{L}}_\mathbf{UF})^k$ become one \cite{Cheng2010b}. Therefore, the diagonal values are used to identify changes in the attractor cycles for different conditions of the inputs and faults.

Note that ${\rm{A}}\vartriangle {\rm{B}}$ is the symmetric difference between the matrices ${\rm{A}}$ and ${\rm{B}}$. $Tr({\rm{A}})$ is the trace of a matrix ${\rm{A}}$.
\begin{thm}:  \label{theorem:7.1}
\textit{Existence theorem}\\
For a given input $\widetilde{\mathbf{U}} \in \widehat{\mathcal{U}}$ in BCN $\rm{L}$ (\ref{eq:Struc2_FB}),  the existence of a fault vector $\mathbf{F}_i$ is ensured iff
\begin{equation*}
Tr  \Big( ({\rm{L}}_{\widetilde{\mathbf{U}} \mathbf{F}_i})^k \vartriangle ({\rm{L}}_{\widetilde{\mathbf{U}} \mathbf{F}_0})^k \Big)  \ne 0 \text{ for any } k \in \{l_1,l_2,\hdots,l_s\},
\end{equation*} 
where ${\rm{L}}_{\widetilde{\mathbf{U}} \mathbf{F}_i}$ indicates resultant structure matrix $\rm{L}$ in the presence of input $\widetilde{\mathbf{U}}$ and fault $\mathbf{F}_i$, $l_p$ is the length of a cycle in BCN ${\rm{L}}_\mathbf{UF}$\big ($l_p < l_{p+1} ; \forall p \in \{1,\hdots,s\}$\big), and $l_s$ is the length of a largest cycle among all ${\rm{L}}_\mathbf{UF}$. 
\end{thm}

\begin{proof}:
\textit{Necessary condition:} For the specified input $\widetilde{\mathbf{U}} \in \widehat{\mathcal{U}}$, drug $\mathbf{D}_0$ and the arbitrary fault $\mathbf{F}_i$ in the BCN $\rm{L}$, the state dynamics of the system is given by: 
\begin{eqnarray*}
	   \mathbf{X}(\tau) = {\rm{L}}\widetilde{\mathbf{U}}\mathbf{F}_i{\mathbf{D}_0}\mathbf{X}({\tau-1}) ={\rm{L}}_{\widetilde{\mathbf{U}} \mathbf{F}_i} \mathbf{X}({\tau-1}),
\end{eqnarray*}
where ${\rm{L}}_{\widetilde{\mathbf{U}} \mathbf{F}_i} = {\rm{L W}}_{[2^\lambda,2^\alpha 3^\gamma]}\mathbf{D}_0\widetilde{\mathbf{U}}\mathbf{F}_i$. $\rm{L}_{\widetilde{\mathbf{U}} \mathbf{F}_i}$ is a matrix of dimension $2^N \times 2^N$, where $N$ is the number of feedback nodes in the BCN. If the system $\rm{L}$ is fault-free, the state equation is 
\begin{eqnarray*}
	   \mathbf{X}(\tau) = {\rm{L}}\widetilde{\mathbf{U}}\mathbf{F}_0{\mathbf{D}_0}\mathbf{X}({\tau-1}) ={\rm{L}}_{\widetilde{\mathbf{U}} \mathbf{F}_0} \mathbf{X}({\tau-1}),
\end{eqnarray*}
where all input conditions are kept identical. 

${\rm{L}}_{\widetilde{\mathbf{U}} \mathbf{F}_i}$ and ${\rm{L}}_{\widetilde{\mathbf{U}} \mathbf{F}_0}$ represent state transition matrices of the network.  The diagonal of matrix   $({\rm{L}}_{\widetilde{\mathbf{U}} \mathbf{F}_i})^k$ represents states involved in cycles of length $k$ and its proper factors. Therefore estimation of $Tr  \Big( ({\rm{L}}_{\widetilde{\mathbf{U}} \mathbf{F}_i})^k \vartriangle ({\rm{L}}_{\widetilde{\mathbf{U}} \mathbf{F}_0})^k \Big)$ for increasing $k \in \{l_1,l_2,\hdots,l_s\}$ shows mismatch in the cycles. If $Tr  \Big( ({\rm{L}}_{\widetilde{\mathbf{U}} \mathbf{F}_i})^k \vartriangle ({\rm{L}}_{\widetilde{\mathbf{U}} \mathbf{F}_0})^k \Big) = 0$, difference between cycles cannot be estimated, and faults $\mathbf{F}_i$ and $\mathbf{F}_0$ become indistinguishable. If any value of $k$ makes $Tr  \Big( ({\rm{L}}_{\widetilde{\mathbf{U}} \mathbf{F}_i})^k \vartriangle ({\rm{L}}_{\widetilde{\mathbf{U}} \mathbf{F}_0})^k \Big) \neq 0$, fault $\mathbf{F}_i$ becomes distinguishable. Therefore, the necessary condition for existence of fault $\mathbf{F}_i$ is $Tr  \Big( ({\rm{L}}_{\widetilde{\mathbf{U}} \mathbf{F}_i})^k \vartriangle ({\rm{L}}_{\widetilde{\mathbf{U}} \mathbf{F}_0})^k \Big) \neq 0  \text{ for any } k \in \{l_1,l_2,\hdots,l_s\}$.\\
{\it Sufficient condition:}
Suppose, 
\begin{align*}
T_r  \Big( ({\rm{L}}_{\widetilde{\mathbf{U}} \mathbf{F}_i})^k \vartriangle ({\rm{L}}_{\widetilde{\mathbf{U}} \mathbf{F}_0})^k \Big) \neq 0  \text{ for any } k \in \{l_1,l_2,\hdots,l_s\},
\end{align*}
where ${\rm{L}}_{\widetilde{\mathbf{U}} \mathbf{F}_i} = {\rm{L W}}_{[2^\lambda,2^\alpha 3^\gamma]}\mathbf{D}_0\widetilde{\mathbf{U}}\mathbf{F}_i$. The above inequality states that attractor cycles produced in presence of fault $\mathbf{F}_i$ and that of fault-free network are not same. Therefore it trivially establishes $\mathbf{F}_i \neq \mathbf{F}_0$. Hence it is sufficient to say that the inequality guaranties the existence of fault $\mathbf{F}_i$.
This completes the proof.
\end{proof}

\begin{exmp}:
Let number of feedback nodes ($N$) inside the network be 3. Therefore, dimensions of matrix  ${\rm{L}}_{\widetilde{\mathbf{U}} \mathbf{F}_i}$ or ${\rm{L}}_{\widetilde{\mathbf{U}} \mathbf{F}_0}$ become $8 \times 8$. Assume that for some network, known input condition $\widetilde{\mathbf{U}}$ and fault $\mathbf{F}_i$, the structure matrix is reduced to
\begin{align*}
{\rm{L}}_{\widetilde{\mathbf{U}} \mathbf{F}_i} &= \delta_8[3,4,6,1,1,6,8,7] \\
\text{ and } {\rm{L}}_{\widetilde{\mathbf{U}} \mathbf{F}_0} &= \delta_8[5,4,6,3,1,6,2,3]
\end{align*}
For $k = 1$, $Tr  \Big( {\rm{L}}_{\widetilde{\mathbf{U}} \mathbf{F}_i} \vartriangle {\rm{L}}_{\widetilde{\mathbf{U}} \mathbf{F}_0} \Big) = 0$. \\
For $k = 2$, $Tr  \Big( ({\rm{L}}_{\widetilde{\mathbf{U}} \mathbf{F}_i})^2 \vartriangle ({\rm{L}}_{\widetilde{\mathbf{U}} \mathbf{F}_0})^2 \Big) \neq 0$. This ensures existence of fault $\mathbf{F}_i$. \exclose
\end{exmp}
\begin{cor}: \label{theorem:7.2}
For a given fault  $\widetilde{\mathbf{F}} \in \widehat{\mathcal{F}}$ in BCN $\rm{L}$ (\ref{eq:Struc2_FB}),  there exists an input $\mathbf{U}_j \in \widehat{\mathcal{U}}$ capable of detecting fault $\widetilde{\mathbf{F}}$ iff 
\begin{equation*}
Tr  \Big( ({\rm{L}}_{\widetilde{\mathbf{F}} \mathbf{U}_j})^k \vartriangle ({\rm{L}}_{\mathbf{F}_0 \mathbf{U}_j})^k \Big)  \ne 0 \text{ for any } k \in \{l_1,l_2,\hdots,l_s\}, \\
\end{equation*} 
where ${\rm{L}}_{\widetilde{\mathbf{F}} \mathbf{U}_j}$ indicates resultant structure matrix $\rm{L}$ in presence of input $\mathbf{U}_j$, fault $\widetilde{\mathbf{F}}$, $l_p$ is the length of a cycle in BCN ${\rm{L}}_\mathbf{FU}$\big ($l_p < l_{p+1} ; \forall p \in \{1,\hdots,s\}$\big), and $l_s$ is the size of a largest cycle among all  ${\rm{L}}_\mathbf{FU}$.
\end{cor}

Assume that the cardinality of the permissible input set $\whsv{U}$ is $\kappa$. A sequence $\pi$ of inputs is a subset of $\whsv{U}$. Therefore, $card(\pi) \leq  \kappa$. The detection of the fault and design of intervention is more reliable when multiple input vectors are available. Hence, the fault analysis and intervention procedures are derived from the sequence of inputs.
 
\begin{cor}:  \label{cor:thm7.2}
For a given fault  $\widetilde{\mathbf{F}} \in \widehat{\mathcal{F}}$ in BCN $\rm{L}$ (\ref{eq:Struc2_FB}),  there exists a sequence $\pi$ ($card(\pi) \leq  \kappa$) of inputs $\mathbf{U}_j \in \widehat{\mathcal{U}} \,\,(j = 1,\hdots,\kappa)$ capable of detecting fault $\widetilde{\mathbf{F}}$ iff  
\begin{equation*}
Tr  \Bigg(  \underset{j=1}{\stackrel{\kappa}{\prod}} {\rm{L}}_{\mathbf{U}_j \widetilde{\mathbf{F}}} \vartriangle \underset{j=1}{\stackrel{\kappa}{\prod}} {\rm{L}}_{\mathbf{U}_j \mathbf{F}_0} \Bigg)  \ne 0,
\end{equation*} 
where ${\rm{L}}_{\mathbf{U}_j \widetilde{\mathbf{F}}}$ indicates the resultant structure matrix $\rm{L}$ in presence of input $\mathbf{U}_j$ and fault $\widetilde{\mathbf{F}}$. Note that the sequence of size $card(\pi)$ is sufficient to detect the fault, and full sequence $\{\mathbf{U}_j, j =1,\hdots, \kappa\}$ may not be required. 
\end{cor}
\begin{thm}:  \label{theorem:8}
\textit{Uniqueness theorem}\\
For a specific fault $\widetilde{\mathbf{F}} \in \widehat{\mathcal{F}}$ in BCN $\rm{L}$ (\ref{eq:Struc2_FB}), input $\widetilde{\mathbf{U}} \in \widehat{\mathcal{U}}$ can uniquely identify fault $\widetilde{\mathbf{F}}$ iff 
\begin{align*}
& Tr  \Big( ({\rm{L}}_{\widetilde{\mathbf{U}} \widetilde{\mathbf{F}}})^k \vartriangle ({\rm{L}}_{\widetilde{\mathbf{U}} \mathbf{F}_0})^k \Big)  \ne 0 \text{ for any } k \in \{l_1,\hdots,l_s\}, \\
\text{ and } & Tr  \Big( ({\rm{L}}_{\widetilde{\mathbf{U}} \mathbf{F}_j})^k \vartriangle ({\rm{L}}_{\widetilde{\mathbf{U}} \mathbf{F}_0})^k \Big)  = 0; \quad \forall k \in \{l_1,l_2,\hdots,l_s\}; \\
&\hspace{1.87in}\forall \mathbf{F}_j \ne \widetilde{\mathbf{F}}, \mathbf{F}_j \in \widehat{\mathcal{F}}.
\end{align*}
\end{thm}

\begin{thm}:  \label{theorem:11}
\textit{Existence of control}\\
For BCN of (\ref{eq:Struc2_FB}) with fault vector $\widetilde{\mathbf{F}} \in {\mathcal{F}}$ and sequence $\pi$ of control inputs, a drug vector $\mathbf{D}_i \in {\mathcal{D}}$ exists if
\begin{equation*}
Tr  \Bigg(  \underset{j=card(\pi)}{\stackrel{1}{\prod}} {\rm{L}}_{\widetilde{\mathbf{F}}} \mathbf{U}_j \mathbf{D}_i \vartriangle \underset{j=card(\pi)}{\stackrel{1}{\prod}} {\rm{L}}_{\mathbf{F}_0} \mathbf{U}_j \mathbf{D}_0 \Bigg)  = 0,
\end{equation*}
where ${\rm{L}}_{\widetilde{\mathbf{F}}}$ indicates the resultant structure matrix ${\rm{L}}$ in presence of fault $\widetilde{\mathbf{F}}$.
\end{thm}
%
%
Other result like the generalized uniqueness, improvement in observability, and improvement in controllability can be easily modified for BCN. Thus, those results for BCN are not shown in this manuscript. 

\begin{exmp}:
{\rm{
Boolean equations for p53 pathways are given as \cite{layek2011b} : 
\begin{align*}
ATM_{next} &= \overline{Wip1} \big(ATM + DNA\_dsb \big), \\
p53_{next} &= \overline{Mdm2} \big(ATM + Wip1 \big), \\
Wip1_{next} &= p53, \\
Mdm2_{next} &= \overline{ATM}\big(p53 + Wip1\big).
\end{align*}
In these pathways, when p53 is `active', it acts as a tumor suppressor (plausible fault location $f_1$) and Mdm2 is one of the target sites for application of drugs ($d_1$). Therefore, $\alpha = 1$, $\lambda = 1$, $\gamma = 1$, and $N = 4$. Structure matrix $\rm{L}$ for the above equations is estimated as shown.
\begin{align*}
{\rm{L}} &= \delta_{16}\big[10,10,2,2,10,10,2,2,9,9,5,5,9,9,5,5,| 14,10, \\
&\hspace{1cm} 6,2,14,10,6,2,13,9,5,5,13, 9,5,5,| 12,12,4,4,\\
&\hspace{1cm} 12,12,4,4,11,11,8,8,11,11,8,8, | 16,12,8,4,16,\\ 
&\hspace{1cm} 12,8,4,15,11,8,8,15, 11,8,8, | 10,10,2,2,12,12,\\
&\hspace{1cm} 4,4,9,9,5,5,11,11,8,8, | 14,10,6,2,16,12,8,4, \\
&\hspace{1cm} 13,9,5,5,15,11,8,8, | 10,10,2,2,10,10,2,2,9,9, \\
&\hspace{1cm} 13,13,9,9,13,13, | 14,10,6,2,14,10,6,2,13,9,13,\\
&\hspace{1cm} 13,13, 9,13,13, | 12,12,4,4,12,12,4,4,11,11,16, \\
&\hspace{1cm} 16,11,11,16,16, | 16,12,8,4,16,12,8,4,15,11,16, \\
&\hspace{1cm} 16,15,11,16,16, | 10,10,2,2,12,12,4,4,9,9,13,\\
&\hspace{1cm} 13,11,11,16,16, | 14,10,6,2,16,12,8,4,13,9,13,\\
&\hspace{1cm} 13,15,11,16,16 \big].
\end{align*}
p53 having `stuck-at 0' results in proliferation. Therefore, it is important to detect this fault. \\
\textit{Existence theorem}: \\
Case 1 : Let $\mathbf{D} = \mathbf{D}_0$, $\widetilde{\mathbf{U}} = \delta_2^2$, and $\mathbf{F}_i = \delta_3^2$ (`stuck-at 0'). Structure matrix is reduced to
\begin{align*}
{\rm{L}}_{\widetilde{\mathbf{U}}\mathbf{F}_i} &= \delta_{16}[16,12,8,4,16,12,8,4, \\ 
&\hspace{1cm} 15,11,16,16,15,11,16,16],\\
{\rm{L}}_{\widetilde{\mathbf{U}}\mathbf{F}_0} &= \delta_{16}[14,10,6,2,16,12,8,4, \\
&\hspace{1cm} 13,9,13,13,15,11,16,16], \\
\text{ and } k &= \{1\}.
\end{align*}
For $k = 1$, $Tr \big( {\rm{L}}_{\widetilde{\mathbf{U}}\mathbf{F}_i} \vartriangle {\rm{L}}_{\widetilde{\mathbf{U}}\mathbf{F}_0} \big) = 1$. This ensures the existence of fault $\mathbf{F}_i = \delta_3^2$. \\
Case 2 : Let $\widetilde{\mathbf{U}} = \delta_2^1$. Keeping other parameters same, structure matrix is reduced to
\begin{align*}
{\rm{L}}_{\widetilde{\mathbf{U}}\mathbf{F}_i} &= \delta_{16}[16,\,12,\,8,\,4,\,16,\,12,\,8,\,4,\, \\
&\hspace{1cm} 15,\,11,\,8,\,8,\,15,\,11,\,8,\,8],\\
{\rm{L}}_{\widetilde{\mathbf{U}}\mathbf{F}_0} &= \delta_{16}[14,\,10,\,6,\,2,\,16,\,12,\,8,\,4,\, \\
&\hspace{1cm} 13,\,9,\,5,\,5,\,15,\,11,\,8,\,8], \\
\text{ and } k &= \{1,7\}.
\end{align*}
For $k = 1$, $Tr \big( {\rm{L}}_{\widetilde{\mathbf{U}}\mathbf{F}_i} \vartriangle {\rm{L}}_{\widetilde{\mathbf{U}}\mathbf{F}_0} \big) = 1$.  This ensures the existence of fault $\mathbf{F}_i = \delta_3^2$. \\
\textit{Improvement in observability} :  Improvement is not required as `stuck-at 0' fault is detectable. \\
\textit{Existence of intervention}: \\
Case 1: Let $\mathbf{U}_j = \delta_2^2$, and fault $\widetilde{\mathbf{F}} = \delta_3^2$ (`stuck-at 0') is detected. Let $\mathbf{D}_i = \delta_2^1$
\begin{align*}
{\rm{L}}_{\widetilde{\mathbf{F}}}\mathbf{U}_j\mathbf{D}_i &= \delta_{16}[12,12,4,4,12,12,4,4, \\
&\hspace{1cm} 11,11,16,16,11,11,16,16],\\
{\rm{L}}_{\mathbf{F}_0}\mathbf{U}_j\mathbf{D}_0 &= \delta_{16}[14,10,6,2,16,12,8,4, \\
&\hspace{1cm} 13,9,13,13,15,11,16,16], \\
\text{ and } k &= \{1\}.
\end{align*}
For $k =1$, $Tr \big({\rm{L}}_{\widetilde{\mathbf{F}}}\mathbf{U}_j\mathbf{D}_i \vartriangle  {\rm{L}}_{\mathbf{F}_0}\mathbf{U}_j\mathbf{D}_0 \big) = 1$. Therefore, drug $\mathbf{D}_i = \delta_2^1$ is not useful. \\
Case 2: Let $\mathbf{U}_j = \delta_2^1$, and fault $\widetilde{\mathbf{F}} = \delta_3^2$ (`stuck-at 0') is detected. Let $\mathbf{D}_i = \delta_2^1$
\begin{align*}
{\rm{L}}_{\widetilde{\mathbf{F}}}\mathbf{U}_j\mathbf{D}_i &= \delta_{16}[12,12,4,4,12,12,4,4, \\
&\hspace{1cm} 11,11,8,8,11,11,8,8],\\
{\rm{L}}_{\mathbf{F}_0}\mathbf{U}_j\mathbf{D}_0 &= \delta_{16}[14,10,6,2,16,12,8,4, \\
&\hspace{1cm} 13,9,5,5,15,11,8,8],\\
\text{ and } k &= \{1,7\}.
\end{align*}
For $k = \{1,7\}$, $Tr \big( ({\rm{L}}_{\widetilde{\mathbf{F}}}\mathbf{U}_j\mathbf{D}_i)^k \vartriangle  ({\rm{L}}_{\mathbf{F}_0}\mathbf{U}_j\mathbf{D}_0)^k \big) \neq 0$. Therefore, drug $\mathbf{D}_i = \delta_2^1$ is not useful. \\
\textit{Improvement in controllability}: Let the target for new drug $d_2$ be $ATM$, $p53$, or $Wip1$.  Applying Algorithm~2, it is observed that target $ATM$ shows some improvement in controllability. When $DNA\_dsb = 0$ and $p53$ is stuck-at 0, for target $ATM$, results obtained are as follows:\\
Assume ${\widetilde{\mathbf{F}}} = \delta_3^2$ and $\mathbf{U}_j = \delta_2^2$. Let $\mathbf{D}_3 = \delta_4^1$, $\mathbf{D}_2 = \delta_4^2$, $\mathbf{D}_1 = \delta_4^3$. Therefore,
\begin{align*}
{\rm{L}}_{\widetilde{\mathbf{F}}}\mathbf{U}_j\mathbf{D}_3 &= \delta_{16}[ 11,11,16,16,11,11,16,16, \\
&\hspace{1cm} 11,11,16,16,11, 11, 16,16],\\
{\rm{L}}_{\widetilde{\mathbf{F}}}\mathbf{U}_j\mathbf{D}_2 &= \delta_{16}[ 12,12,4,4,12,12,4,4, \\
&\hspace{1cm} 11,11,16,16,11,11,16,16],\\
{\rm{L}}_{\widetilde{\mathbf{F}}}\mathbf{U}_j\mathbf{D}_1 &= \delta_{16}[15,11,16,16,15,11,16,16, \\
&\hspace{1cm} 15,11,16,16,15,11,16,16],\\
{\rm{L}}_{\mathbf{F}_0}\mathbf{U}_j\mathbf{D}_0 &= \delta_{16}[14,10,6,2,16,12,8,4, \\
&\hspace{1cm} 13,9,13,13,15,11,16,16], \\
\text{ and } k &= \{1\}.
\end{align*}
For $k = 1$,
\begin{align*}
T_r \big({\rm{L}}_{\widetilde{\mathbf{F}}}\mathbf{U}_j\mathbf{D}_3 \vartriangle  {\rm{L}}_{\mathbf{F}_0}\mathbf{U}_j\mathbf{D}_0 \big) &= 0, \\
T_r \big({\rm{L}}_{\widetilde{\mathbf{F}}}\mathbf{U}_j\mathbf{D}_2 \vartriangle {\rm{L}}_{\mathbf{F}_0}\mathbf{U}_j\mathbf{D}_0 \big) &= 1, \\
T_r \big({\rm{L}}_{\widetilde{\mathbf{F}}}\mathbf{U}_j\mathbf{D}_1 \vartriangle {\rm{L}}_{\mathbf{F}_0}\mathbf{U}_j\mathbf{D}_0 \big) &= 0. \\
\end{align*}
Therefore drugs $\mathbf{D}_3$ and $\mathbf{D}_1$ are useful interventions. In logical equivalence, $D_3 \sim (1\,1)$ and $D_1 \sim (0\,1)$. Therefore, only drug $d_2$ is effective. This shows the effectiveness of the method in discarding certain drugs and selecting an appropriate therapeutic intervention. A suitable inhibitory drug at ATM may be useful for the stuck-at 0 fault at p53. This prediction shows that the method may prove to be helpful to decide the future research towards drugs discovery.}}  \exclose
\end{exmp}
\section{Conclusion}
The manuscript describes a linear approach towards fault analysis and intervention in Boolean systems. The methodology opens up new problems towards fault analysis and intervention in Boolean systems.
The proposed study considers the possibility of multiple faults (mutations) in feedback networks. The method does not require any test set considering the experimental difficulty in assigning test inputs to biological networks. 

The objective of this work is to obtain the optimal therapeutic intervention with the available input-output information. In some cases, extra reporters may be required to analyze the mutations. If the none of drugs are useful from the available set, the improvement in controllability procedure suggests the new possible targets for the drugs. Although the method is exponential, the drug estimation time is still less than the treatment time of the patient. Also, the drugs obtained with the proposed work can be used to improve the lifespan of a patient and save the experimentation cost. Future work can be on the output based fault identification and control of BCN and the corresponding improvement in observability and controllability which have been described in this study. Fault analysis and control in the paradigms of asynchronous Boolean networks and probabilistic Boolean networks can also be taken into consideration in the near future.
\section*{Acknowledgment}
The work was supported by the project ‘WBC’ funded by Indian Institute of Technology, Kharagpur.


\bibliographystyle{IEEEtran}

\bibliography{reference}

\vfill

\newpage
\appendix
\section{Preliminaries} \label{supsec:intro}
\subsection{Generalization of matrices}
\begin{lem}: \label{lem:power_red}
\textit{Power reduction matrix}\\
Let $\tv{a}$ be the column vector of dimension $p$ and $\tv{A} = \tv{a}_1\tv{a}_2\cdots \tv{a}_k$. Then $\tv{A}^2 = \Phi_p \tv{A}$, where 
\begin{equation*}
\Phi_p = \delta_{p^2}\big[1,1+(p+1),1+2(p+1),\hdots,1+(p-1)(p+1)\big].
\end{equation*}
\end{lem}

\begin{lem}: \label{lem:dummy}
\textit{Dummy matrix} \\
For column vectors $\tv{A} \in \mathbb{R}^p$ and $\tv{B} \in \mathbb{R}^q$, dummy matrix $\mv{E}_{[p,q]}$ is defined as
\begin{align*}
\mv{E}_{[p,q]} &= [\underbrace{\mv{I}_p \, \mv{I}_p \hdots \,\mv{I}_p}_{q \text{ times }}],
\end{align*}
such that $\tv{A} = \mv{E}_{[p,q]}\tv{B\,A}$ or $\tv{A} = \mv{E}_{[p,q]}\mv{W}_{[p,q]}\tv{A\,B}$. When $p = q$, it can be simply represented as $\mv{E}_{[p]}$.
\end{lem}

\begin{defn} (see \cite{Cheng2010b})
If $\tv{A}$ is a vector of dimension $k$ and $\tv{B} = \mv{M}_1\tv{A}\mv{M}_2\tv{A}\cdots$ $\mv{M}_q\tv{A}$, then $\tv{B}$ can be written in simplified form as $\tv{B} = \mv{M}^\ast \tv{A}$, where
\begin{equation*}
\mv{M}^\ast = \mv{M}_1 \, \underset{i=2}{\stackrel{q}{\prod}} \Big[\big( \mv{I}_{k} \otimes \mv{M}_i)\Phi_k \Big].
\end{equation*}
\end{defn}

\subsection{Derivation of structure matrix for faults} \label{supsec:faultdrug}
Assume that the fault is expected to strike at some location $x$ which changes its value to $x^\ast$. As $x \in \{0,1\}$, $\tv{x} \in \Delta_2$. Similarly, $\tv{x}^\ast \in \Delta_2$. Using STP, this condition can be written in linear form as:
\begin{align}
\tv{x}^\ast &= \mv{M}_f \tv{x\,f}.	\label{eq:fault_intervention}
\end{align}
In this expression, $\mv{M}_f \in \sv{B}_{2 \times 6}$ is the structure matrix for the fault model. For $x = 1$ and sa-1 appears at $f$, then $\tv{x} = \delta_2^1$ and $\tv{f} = \delta_3^1$. Hence, $\tv{x\,f} = \delta_6^1$. For this combination of $x$ and $f$, the value $x^\ast = 1$. Therefore, $\mv{M}_f \delta_6^1 = \delta_2^1$ and $\delta_2^1$ gives the first column of the structure matrix $\mv{M}_f$. Similarly, $\mv{M}_f$ can be constructed from all combinations of $x$ and $f$. This process gives the structure matrix $\mv{M}_f= \delta_2[1, 2, 1, 1, 2, 2]$.
%
%

\subsection{Derivation of structure matrix for drugs}
Assume that a drug $d$ be applied at location $x$ and its value is modified to $x^\ast$ on the application of the drug. The equation for drug intervention in logical form is given as:
\[
x^\ast  = x \wedge (\neg d).
\]
This equation can be written in linear form as:
\begin{align}
\tv{x}^\ast =  \mv{M}_c \tv{x} (\mv{M}_n \tv{d}) = \mv{M}_c (\mv{I}_2 \otimes \mv{M}_n) \tv{x \, d} = \mv{M}_D \tv{x \, d}. \label{eq:drug_intervention}
\end{align}
In this expression, $\mv{M}_D \in \sv{B}_{2 \times 4}$ is the structure matrix for the drug intervention, $\mv{I}_2$ is an identity matrix of size $2 \times 2$. $\mv{M}_c = \delta_2[1,2,2,2]$ and $\mv{M}_n = \delta_2[2,1]$ are the structure matrices for logical AND operation and logical NOT operation respectively. These matrices are available in \cite{Cheng2010b}. The structure matrix $\mv{M}_D$ is given by:
\begin{equation*}
\mv{M}_D = \mv{M}_c (\mv{I}_2 \otimes \mv{M}_n) = \delta_2[2, 1, 2, 2].
\end{equation*}

\section{Derivation of structure matrix for Boolean maps} \label{supsec:BMder}
In the manuscript, Fig.~1 represents the block diagram of Boolean control networks (BCN). For Boolean Maps (BM), state feedback link is removed. Procedure to estimate structure matrix $\mv{H}$ in equation~(8) of the manuscript is as follows:

Output at level $i$ ($X^i$) in primary block is dependent of overall structure matrix of that level (let us call it $\mv{L}^i$), and external inputs $U$, present faults $F$, and drugs $D$. In linear form, assume that output $\tv{X}^i$ given by:
\begin{align}
\tv{X}^i &= \mv{L}^i \tv{U F D}. \nonumber
\end{align}
From~Fig.~1 in manuscript, equation of individual output for each block $j$ in level $i$ can be written as:
\begin{align}
\tv{x}^i_j &= \mv{L}^i_j \tv{UFDX}^1\cdots\,\tv{X}^{i-1},  \label{eq:XijGeneral}
\intertext{where $i =1,\hdots,m$ and $j = 1,\hdots,n_i$. Using earlier assumption,}
\tv{x}^i_j &= \mv{L}^i_j \tv{UFD} (\mv{L}^1\tv{UFD})\cdots (\mv{L}^{i-1} \tv{U F D}). \nonumber
\intertext{From Lemma~\ref{lem:power_red}, $(\tv{UFD})^2 = \Phi_\omega \tv{UFD}$ and $\omega = 2^{\alpha+\lambda}3^\gamma$. Therefore,}
\tv{x}^i_j	 & = \mv{L}^i_j \, \underset{k=1}{\stackrel{i-1}{\prod}} \Big[\big( \mv{I}_{\omega} \otimes \mv{L}^k \big)\Phi_\omega \Big] \tv{UFD}  = \mv{L}^{i\ast}_j \tv{UFD}, \label{eq:Xij} \\
\text{where } 
\mv{L}^{i\ast}_{j} & = \mv{L}^i_j \, \underset{k=1}{\stackrel{i-1}{\prod}} \Big[\big(\mv{I}_{\omega} \otimes \mv{L}^k \big)\Phi_\omega \Big]. \nonumber \\
\intertext{Output vector at level $i$ is given by}
\tv{X}^i &= \tv{x}^i_1\,\tv{x}^i_2\,\cdots\,\tv{x}^i_{n_i} \nonumber \\
	   &= (\mv{L}^{i\ast}_1\tv{UFD})\,(\mv{L}^{i\ast}_2\tv{UFD})\,\cdots\,(\mv{L}^{i\ast}_{n_i}\tv{UFD})  \nonumber \\
     &= \mv{L}^i \tv{UFD}, \label{eq:vect_opXi}
\end{align}
where $\mv{L}^i = \mv{L}^{i\ast}_1\, \underset{j=2}{\stackrel{n_i}{\prod}} \Big[\big(\mv{I}_{\omega} \otimes \mv{L}^{i\ast}_j \big)\Phi_\omega\Big] \nonumber$.

The resultant output vector of all the $m$ levels in primary block is STP of the outputs of individual levels in that block. Therefore,
\begin{align}
\tv{X} &= \tv{X}^1 \tv{X}^2 \cdots \tv{X}^m. \label{eq:X1}\\
\intertext{Substituting values using equation~(\ref{eq:vect_opXi}),}
\tv{X} &= (\mv{L}^1\tv{UFD})(\mv{L}^2 \tv{UFD})\cdots (\mv{L}^m \tv{UFD}) \nonumber \\
&= \mv{L} \tv{UFD}, \label{eq:X2} \\
\text{where }
\mv{L} &= \mv{L}^1\, \underset{j=2}{\stackrel{m}{\prod}} \Big[\big(\mv{I}_{\omega} \otimes \mv{L}^j \big)\Phi_\omega\Big]. \label{eq:L}
\end{align}

Primary output vector $\tv{Y}$ of the BM depends on the secondary level. Inputs of secondary level are output vector from primary block and primary inputs of BM. Therefore equation of primary output vector is derived as:
\begin{align}
\tv{y}_1 &= \mv{H}_1\tv{UX}. \nonumber \\
\intertext{Substituting $\tv{X}$ from equation~(\ref{eq:X2}),}
\tv{y}_1 &= \mv{H}_1 \tv{U(LUFD)} \nonumber \\
&= \mv{H}_1^\ast \tv{UFD}, \label{eq:PO_single} \\ 
\intertext{where $\mv{H}_1^\ast = \mv{H}_1 \big(\mv{I}_{2^\alpha} \otimes \mv{L} \big)\Phi_{2^\alpha}$. Similarly, other primary outputs are defined as:}
\tv{y}_2 &= \mv{H}_2^\ast \tv{UFD} \nonumber \\
&\hspace{0.2cm}\vdots \nonumber \\
\tv{y}_\beta &= \mv{H}_\beta^\ast \tv{UFD}. \nonumber
\end{align}
Output vector $\tv{Y}$ is given as:
\begin{align}
\tv{Y} &= \tv{y}_1\,\tv{y}_2\,\cdots\,\tv{y}_\beta \nonumber \\
&= (\mv{H}_1^\ast \tv{UFD})(\mv{H}_2^\ast \tv{UFD})\cdots (\mv{H}_\beta^\ast \tv{UFD}) \nonumber \\
&= \mv{H}_1^\ast\, \underset{j=2}{\stackrel{\beta}{\prod}} \Big[\big(\mv{I}_{\omega} \otimes \mv{H}_j^\ast \big)\Phi_\omega\Big] \tv{UFD} \nonumber \\
&= \mv{H}\tv{UFD}, \label{eq:supPO}\\
\intertext{where $\mv{H}= \mv{H}_1^\ast\, \underset{j=2}{\stackrel{\beta}{\prod}} \Big[\big(\mv{I}_{\omega} \otimes \mv{H}_j^\ast \big)\Phi_\omega\Big] \nonumber$.}
\end{align}
\vspace{-2cm}
\section{Derivation for improvement in controllability in Boolean maps} \label{supsec:control}
Let intermediate output $x^m_1$ is target location of drug $d_{\lambda+1}$. Then $\tv{x}^m_1$ is modified as, $\tv{x}^{m\ast}_1 = \mv{M}_D \tv{x}^m_1 \tv{d}_{\lambda+1}$. Therefore from equations (\ref{eq:Xij}) and (\ref{eq:vect_opXi}), output of level $m$ is given by:
\allowdisplaybreaks
\begin{align}
\tv{X}^m &= \tv{x}^{m\ast}_1\, \tv{x}^m_2\,\cdots\,\tv{x}^m_{n_m} \nonumber \\
		 &= (\mv{M}_D \tv{x}^{m}_1 \tv{d}_{\lambda+1})\, \tv{x}^m_2\,\cdots\,\tv{x}^m_{n_m} \nonumber \\
&= (\mv{M}_D\,\mv{L}^{m\ast}_1\,\tv{UFD\,d}_{\lambda+1})\,\mv{L}^{m\ast}_2 \tv{UFD}\cdots\,\mv{L}^{m\ast}_{n_m}\tv{UFD} \nonumber \\
&=  (\mv{M}_D\,\mv{L}^{m\ast}_1\,\tv{UF}\hat{\tv{D}})\,\mv{L}^{m\ast}_2\tv{UFD}\cdots\,\mv{L}^{m\ast}_{n_m}\tv{UFD}. \label{eq:vector_XNew}
\end{align}
Addition of new drug $d_{\lambda+1}$ requires to change STP of input $\tv{UFD}$ as:
\begin{align*}
\tv{UFD} &= \tv{UFd}_1\tv{d}_2\cdots \tv{d}_{\lambda} \\
&=  \tv{UFd}_1 \tv{d}_2\cdots \tv{d}_{\lambda-1}\mv{E}_{[2]}\mv{W}_{[2]}\tv{d}_{\lambda}\tv{d}_{\lambda+1} \\
&= \big(\mv{I}_{\frac{\omega}{2}} \otimes \mv{E}_{[2]}\mv{W}_{[2]}\big) \tv{UFd}_1 \tv{d}_2\cdots \tv{d}_{\lambda+1} \\
&= \big(\mv{I}_{\frac{\omega}{2}} \otimes \mv{E}_{[2]}\mv{W}_{[2]} \big)\tv{UF}\hat{\tv{D}}. 
\end{align*}
Substituting in equation~(\ref{eq:vector_XNew}),
\begin{align}
\tv{X}^m &= (\mv{M}_D\,\mv{L}^{m\ast}_1\,\tv{UF}\hat{\tv{D}})\,\mv{L}^{m\ast}_2\big(\mv{I}_{\frac{\omega}{2}} \otimes \mv{E}_{[2]} \mv{W}_{[2]}\big) \tv{UF}\hat{\tv{D}} \nonumber \\
&~\quad \cdots\,\mv{L}^{m\ast}_{n_m}\big(\mv{I}_{\frac{\omega}{2}} \otimes \mv{E}_{[2]}\mv{W}_{[2]} \big)\tv{UF}\hat{\tv{D}} \nonumber \\
&= (\hat{\mv{L}}^{m\ast}_1 \tv{UF}\hat{\tv{D}})(\hat{\mv{L}}^{m\ast}_2 \tv{UF}\hat{\tv{D}})\cdots (\hat{\mv{L}}^{m\ast}_{n_m}\tv{UF}\hat{\tv{D}}) \nonumber \\
&= \hat{\mv{L}}^{m\ast}_1 \underset{j=2}{\stackrel{n_m}{\prod}} \Big[\big(\mv{I}_{2\omega} \otimes \hat{\mv{L}}^{m\ast}_j \big)\Phi_{2\omega}\Big] \tv{UF}\hat{\tv{D}} \nonumber \\
&= \hat{\mv{L}}^m \tv{UF}\hat{\tv{D}}, \nonumber
\end{align}
where
\begin{equation*}
\hat{\mv{L}}^{m\ast}_i =
\begin{cases}
\mv{M}_D\,\mv{L}^{m\ast}_i, &  i = 1 \\
\mv{L}^{m\ast}_i\big(\mv{I}_{\frac{\omega}{2}} \otimes \mv{E}_{[2]}\mv{W}_{[2]} \big), & \text{otherwise} 
\end{cases}
\end{equation*}
and $\hat{\mv{L}}^m = \hat{\mv{L}}^{m\ast}_1 \underset{j=2}{\stackrel{n_m}{\prod}} \Big[\big(\mv{I}_{2\omega} \otimes \hat{\mv{L}}^{m\ast}_j \big)\Phi_{2\omega}\Big] $.
 
Similarly, any drug location can be selected from outputs of levels in primary block. Generalized form of output vector for level $j$ for target location $x_i^j$ is given by:
\begin{align}
\tv{X}^j &= \tv{x}_1^j \tv{x}_2^j \cdots \tv{x}_{n_i}^j \nonumber \\
 &= \hat{\mv{L}}^{j\ast}_1 \underset{k=2}{\stackrel{n_j}{\prod}} \Big[\big(\mv{I}_{2\omega} \otimes \hat{\mv{L}}^{j\ast}_k \big)\Phi_{2\omega}\Big] \tv{UF}\hat{\tv{D}} \nonumber \\
&= \hat{\mv{L}}^j \tv{UF}\hat{\tv{D}}, \nonumber 
\end{align}
where
\begin{equation*}
\hat{\mv{L}}^{l\ast}_k =
\begin{cases}
\mv{M}_D\,\mv{L}^{l\ast}_k, &  l = j \text{ and } k = i  \\
\mv{L}^{l\ast}_k\big(\mv{I}_{\frac{\omega}{2}} \otimes \mv{E}_{[2]}\mv{W}_{[2]} \big), & l \neq j \text{ or } k \neq i 
\end{cases}
\end{equation*}
and $\hat{\mv{L}}^j = \hat{\mv{L}}^{j\ast}_1 \underset{k=2}{\stackrel{n_j}{\prod}} \Big[\big(  I_{2\omega} \otimes \hat{\mv{L}}^{j\ast}_k \big)\Phi_{2\omega}\Big]$. \\
From equation~(\ref{eq:X1}), equation of output for primary block is given by:
\begin{align}
\tv{X} &= \tv{X}^1\tv{X}^2\cdots \tv{X}^m \nonumber \\
 &= (\hat{\mv{L}}^1\tv{UF}\hat{\tv{D}})(\hat{\mv{L}}^2\tv{UF}\hat{\tv{D}})\cdots (\hat{\mv{L}}^m \tv{UF}\hat{\tv{D}}) \nonumber \\
 &= \hat{\mv{L}}^1 \underset{j=2}{\stackrel{m}{\prod}} \Big[\big(\mv{I}_{2\omega} \otimes \hat{\mv{L}}^{j}\big)\Phi_{2\omega}\Big] \tv{UF}\hat{\tv{D}} \nonumber \\
 &= \widehat{\mv{L}}\tv{UF}\hat{\tv{D}}, \label{eq:X2_new} \\
 \text{where }
 \widehat{\mv{L}} &= \hat{\mv{L}}^1 \underset{j=2}{\stackrel{m}{\prod}} \Big[\big(\mv{I}_{2\omega} \otimes \hat{\mv{L}}^{j}\big)\Phi_{2\omega}\Big]. \label{eq:L_new}
\end{align}

Let us divide the problem in two parts.
\begin{enumerate}
\item \textit{The new drug location is selected from primary block:} The structure matrix of primary block is represented as  equation~(\ref{eq:L_new}). Similar to equations~(\ref{eq:PO_single}) and (\ref{eq:supPO}), net structure matrix of BN can be estimated as: 
\begin{align}
\hat{\mv{H}}^\ast_j &= \mv{H}_j \big(\mv{I}_{2^\alpha} \otimes \widehat{\mv{L}} \big)\Phi_{2^\alpha} \tv{UF}\hat{\tv{D}} \nonumber\\  
\widehat{\mv{H}} &= \hat{\mv{H}}_1^\ast\, \underset{j=2}{\stackrel{\beta}{\prod}} \Big[\big(\mv{I}_{2\omega} \otimes \hat{\mv{H}}_j^\ast \big)\Phi_{2\omega}\Big] \tv{UF}\hat{\tv{D}}\nonumber \\ 
\tv{Y} &= \widehat{\mv{H}} \tv{UF}\hat{\tv{D}}. \label{eq:PO_new1}
\end{align}
\item \textit{The new drug location is selected from primary outputs:}
If target location is selected one of the primary outputs $y_j$, the structure matrix $\mv{L}$ of primary blocks remains unaffected. Therefore from equations~(\ref{eq:PO_single}) and (\ref{eq:supPO}), equation of primary output can be computed as: 
\begin{align}
\hat{\mv{H}}_k^\ast &= 
\begin{cases}
\mv{M}_D \mv{H}_k \big(\mv{I}_{2^\alpha} \otimes \mv{L} \big)\Phi_{2^\alpha}, &  k = j  \\
\mv{H}_k \big(\mv{I}_{2^\alpha} \otimes \mv{L} \big)\Phi_{2^\alpha}\big(\mv{I}_{\frac{\omega}{2}} \otimes \mv{E}_{[2]}\mv{W}_{[2]} \big), & \text{otherwise }
\end{cases} \nonumber
\end{align}
and
\begin{align}
\widehat{\mv{H}}& = \hat{\mv{H}}_1^\ast\, \underset{j=2}{\stackrel{\beta}{\prod}} \Big[\big(\mv{I}_{2\omega} \otimes \hat{\mv{H}}_j^\ast \big)\Phi_{2\omega}\Big] \nonumber \\
\tv{Y} &= \widehat{\mv{H}}\tv{UF}\hat{\tv{D}}. \label{eq:PO_new2}
\end{align}
\end{enumerate}
The final form of equations~(\ref{eq:PO_new1}) and (\ref{eq:PO_new2}) is shown as equation~(10) in the manuscript.
\vfill
\end{document}